\documentclass[a4paper,11pt]{article}
\usepackage[table]{xcolor}
\usepackage[utf8]{inputenc}
\usepackage{amsmath,amssymb,graphicx,pgfplots,enumitem,comment,csquotes,subcaption}
\pgfplotsset{compat=newest}
\usepackage{amsthm}
\usepackage[margin=1in]{geometry}

\usepackage[hidelinks]{hyperref}
\usepackage{cleveref}
\usepackage{algorithm}
\usepackage{tikz}
\usepackage{latexsym}
\usepackage[noend]{algpseudocode}
\usepackage{authblk,booktabs,tabularx}
\usetikzlibrary{cd}
\allowdisplaybreaks

\let\oldbibliography\thebibliography
\renewcommand{\thebibliography}[1]{\oldbibliography{#1}
\setlength{\itemsep}{1pt}}

\newcommand{\la}{\leftarrow}

\newcommand{\hide}[1]{}

\DeclareMathOperator*{\argmin}{arg\,min}
\DeclareMathOperator*{\argmax}{arg\,max}

\newcommand{\bbb}{\mathbf{b}}

\newcommand{\bbp}{\mathbf{p}}

\newcommand{\bbw}{\mathbf{w}}
\newcommand{\bbx}{\mathbf{x}}

\newcommand{\bbz}{\mathbf{z}}

\newcommand{\inner}[2]{\left\langle #1,#2\right\rangle}

\newcommand{\rint}{\mathsf{rint}}
\newcommand{\KL}{\mathrm{KL}}
\newtheorem{example}{Example}
\newtheorem{theorem}{Theorem}

\title{Learning in Markets: Greed Leads to Chaos but \\Following the Price is Right}

\author[1]{Yun Kuen Cheung}
\author[2]{Stefanos Leonardos}
\author[2]{Georgios Piliouras}
\affil[1]{Royal Holloway University of London, \emph{yunkuen.cheung@rhul.ac.uk}}
\affil[2]{Singapore University of Technology and Design, \emph{\{stefanos\_leonardos, georgios\}@sutd.edu.sg}}
\date{}

\makeatletter
\def\BState{\State\hskip-\ALG@thistlm}
\makeatother
\let\oldReturn\Return
\renewcommand{\Return}{\State\oldReturn}
\algnewcommand{\Initialize}[1]{%
  \State \textbf{initialize:} \hspace*{0.2em}\parbox[t]{.4\linewidth}{\raggedright #1}
}

\renewcommand{\(}{\left(}
\renewcommand{\)}{\right)}
\newcommand{\lt}{\left[}
\newcommand{\rt}{\right]}
\newcommand{\qt}{\enquote}

\newcommand{\x}{x_{ij}}
\renewcommand{\b}{b_{ij}}
\renewcommand{\v}{v_{ij}}
\renewcommand{\comment}[1]{\ignorespaces}

\renewcommand{\L}{\mathcal L}
\newcommand{\tb}{\tilde{b}}

\newcommand{\fof}{f^{\(3\)}\(x\)}

\newcolumntype{L}{l>{\hspace*{-0.6\tabcolsep}}}
\newcolumntype{R}{r<{\hspace*{-0.6\tabcolsep}}}
\newcolumntype{C}{c<{\hspace*{-0.6\tabcolsep}}}
\definecolor{scolor}{rgb}{0.733,0.843,0.890}
\definecolor{ocolor}{rgb}{0.783,0.893,0.940}
\definecolor{oocolor}{rgb}{0.783,0.893,0.940}

\definecolor{ppcolor}{rgb}{0.212,0.506,0.584}
\definecolor{pcolor}{rgb}{0.255,0.596,0.686}
\definecolor{rcolor}{rgb}{0.294,0.675,0.776}
\definecolor{ecolor}{rgb}{0.569,0.765,0.835}
\definecolor{scolor}{rgb}{0.733,0.843,0.890}
\definecolor{ocolor}{rgb}{0.783,0.893,0.940}
\definecolor{oocolor}{rgb}{0.783,0.893,0.940}

\newtheorem{lemma}[theorem]{Lemma}

\newtheorem{proposition}[theorem]{Proposition}
\theoremstyle{definition}
\newtheorem{definition}[theorem]{Definition}
\theoremstyle{remark}
\newtheorem*{remark}{Remark}

\begin{document}
\maketitle

\begin{abstract}

We study learning dynamics in distributed production economies such as blockchain mining, peer-to-peer file sharing and crowdsourcing.
These economies can be modelled as multi-product Cour\-not competitions or all-pay auctions (Tullock contests) when individual firms have market power,
or as Fisher markets with quasi-linear utilities when every firm has negligible influence on market outcomes.
In the former case, we provide a formal proof that Gradient Ascent (GA) can be Li-Yorke chaotic
for a step size as small as $\Theta(1/n)$, where $n$ is the number of firms. 
In stark contrast, for the Fisher market case, we derive a Proportional Response (PR) protocol that converges to market equilibrium.
The positive results on the convergence of the PR dynamics are obtained in full generality,
in the sense that they hold for Fisher markets with \emph{any} quasi-linear utility functions.
Conversely, the chaos results for the GA dynamics are established even in the simplest possible setting of two firms and one good,
and they hold for a wide range of price functions with different demand elasticities. Our findings suggest that by considering multi-agent interactions from a market rather than a game-theoretic perspective, we can formally derive natural learning protocols which are stable and converge to effective outcomes rather than being chaotic.
\end{abstract}

\section{Introduction}\label{sec:introduction}
Multi-agent learning in production economies is an important yet underexplored domain. Production economies are classically modelled as Cournot competitions \cite{Varian2010} or imperfectly discriminating all-pay auctions (Tullock contests) \cite{Dip09}. In these models, participating firms have \emph{market power},
and they can significantly influence aggregate outcomes (prices or total exerted effort) with their decisions. However, the advancement of the internet has prompted a rapid paradigm shift in economic competition. Blockchain mining \cite{Arn18,Fia19}, peer-to-peer file sharing~\cite{LLSB08} and crowdsourcing \cite{Hor10}, among others, all constitute \emph{distributed} production economies with large numbers of small competitors (miners, individuals or firms). In contrast to the classic Cournot or Tullock models, firms that participate in these economies typically engage in multiple concurrent competitions. Moreover, due to their relative small sizes, each firm only has negligible influence on prices and hence, become \emph{price-takers}. As a result, this form of competition more closely resembles
the economic model of Fisher markets in which firms take prices as independently given signals, and purchase optimal bundles of goods (or invest on optimal portfolios to produce goods) given their budget (or capital) constraints.\par
The question of which adaptive or learning protocols
behave well in these economies is largely open and still actively researched. In both Cournot competition\footnote{The mathematical equivalence between Cournot competition with isoelastic demand and imperfectly discriminating all-pay auctions with proportional success functions or simply Tullock contests is documented in \cite{Szi97,War18} (among others). We elaborate on this relation in Section~\ref{sec:model}.} and Fisher markets, firms repeatedly observe the aggregate production, and adjust their production outputs over time to improve their own profits. However, empirical results regarding Cournot competition suggest that standard adaptive algorithms, e.g., best response, 
can lead to rather unstable and irregular adjustments, even in very simple instances (e.g., when there are only two firms and one good) \cite{The60,Puu91,War18}. In contrast, when firms ignore their market power and act as price-takers, the outcomes can be more stable. A line of recent works~\cite{WZ2007,Zhang2011,Bir11,CCR12,Cheu19,Che18,BDR19,CheungHN2019} showed that natural adaptive algorithms, including t\^atonnement and proportional response, lead to stable adjustments in many families of Fisher markets, where they converge to \emph{market equilibria}.

\paragraph{Our contribution.}

Motivated by the above, our aim is to study the behavior of learning dynamics in production economies from a theoretical perspective. Our research goals are 1) to establish formal mathematical arguments that explain the irregular behavior of \emph{greedy} learning rules, such as Gradient Ascent and Best Response dynamics, and 2) to seek protocols that behave well under general conditions. In these directions, we make the following two contributions.

Concerning the first goal, we present the first rigorous mathematical proof that the constant step-size \emph{Gradient Ascent (GA)} algorithm can exhibit \emph{Li-Yorke chaos}~\cite{Li75} in Cournot competition (equivalently, in all-pay auctions or Tullock contests) even when the firms are homogeneous. This provides a formal explanation for the \emph{unpredictable} evolution of these systems that is frequently observed in practice. To derive this result, we leverage Sharkovsky's theorem which provides a tractable way to verify the conditions in Li-Yorke's characterization of chaos \cite{Pal17}. In the case of GA, our findings are robust in two aspects: first, chaos emerges for a large family of price functions induced by different demand elasticities, and second, chaos emerges even when the step-size is as small as $\Theta\(1/n\)$. Our results in this direction contribute to the growing literature that studies various forms of chaos in game dynamics. \cite{Sky92,SAF2002,GF2013,CheungP19,CheungP20,CT21,CFMP2020,Bie21,Cho21,Leo21}.

Informally, a dynamical system is Li-Yorke chaotic if there are uncountably many pairs of trajectories which get arbitrarily close together (but never intersect) and move apart indefinitely. When two trajectories are very close to each other, they become essentially indistinguishable due to the precision limitation inherent with the environment or computer. In other words, we cannot tell which of the two trajectories will be realized in the future --- this is exactly what \emph{unpredictable} means. A primary reason for the chaos to arise is that each firm uses its own market power to strategically influence the price. When all firms make such strategic manipulations simultaneously, they aggregately drive prices up and down without proper control. \par
While the previous technique does not lead to a formal proof of Li-Yorke chaos in the case of Best Response (BR) dynamics, we formalize the (in)-stability properties of the latter via eigenvalue analysis of a first-order linear approximation of the non-linear dynamical system. Here, instability refers to abrupt changes in the long term behavior of the dynamics in response to small perturbations of the systems' parameters (e.g., firms costs).\footnote{This formalization closely mirrors the existing empirical results about BR dynamics, see e.g., \cite{Puu91,War18}. Hence, we only retain some intuitive visualizations in the main text (\Cref{fig:spiral_1}), and defer the formal statement to Appendix \ref{app:br}.}\par
Since robustness is an essential property in distributed production economies both from a normative and a descriptive perspective,
the above results provide a convincing argument against the use of \emph{game-theoretically motivated} protocols. This brings us to our second goal which is to seek learning protocols that result in stable outcomes.\par
Our main result in this direction is to propose a \emph{market-motivated} Proportional Response (PR) algorithm and show that it is stable and robust: from any initial condition, the PR update rule converges to the market equilibrium of an ensuing Fisher market that captures production economies, namely Fisher market with quasi-linear utility functions. The protocol is simple and can be run by each firm independently using only local and observable (market level) information, which makes it particularly suitable for these distributed settings. It can be interpreted as a naturally motivated adaptive algorithm from a firm's perspective: in each round, each firm appropriates a certain amount of money, and invests it to the productions of different goods in proportion to the revenues received from selling them in the previous round.\par
One necessary assumption to establish this result is that as economies grow larger, firms have a negligible influence on aggregate outputs (prices or total exerted efforts). However, we formally argue that in the distributed production economy setting, market equilibria are approximate Nash equilibria. This finding is in line with the largeness concept in \cite{Col16}, who showed that when markets grow \emph{large}, they become asymptotically efficient even under agents' strategic behaviors. This implies that the assumption of diminished influence on outcomes does not significantly affect the equilibrium outcome of the system.
However, it does have important implications from a technical perspective. In particular, by modeling production economies as Fisher markets, we can leverage their \emph{Eisenberg-Gale convex-program formulation} \cite{EG59,Eisenberg61} to draw direct analogue between our PR algorithm and standard optimization methods like mirror descent. This allows us to apply tools from optimization theory and provides a principled approach to derive proofs of convergence.

\paragraph{Other Related Work.} The Cournot model dates back to the early 19th century~\cite{Cournot1838}.
Since then, it became a foundation for many models of production economies \cite{Varian2010}. As game theory subsequently matured, the competition between firms was revisited via the popular perspective of Nash equilibrium \cite{Na51}. The mathematical equivalence between Cournot competition and imperfectly discriminating all-pay auctions with proportional success functions or simply Tullock contests \cite{Tul80} when price functions are isoelastic is documented in \cite{Szi97,War18} (among others). The study of markets is also among the most classical topics in Economics, dating back to~\cite{Walras}. Many markets have efficient equilibria\footnote{This was made as a hypothesis by Adam Smith, popularly referred as ``the Invisible Hand''. The famous Fundamental Theorem of Welfare Economics confirmed it analytically.}, which is not true in games. But their assumption of non-strategic (i.e., price-taking) behavior of agents is consistently being challenged. Some recent works remedy this by injecting strategic considerations into the market models~\cite{Adsul2010,Chen2011,Chen2012,Babaioff2014,Branzei2014}. \cite{Col16} showed that when markets grow \emph{large}, they become asymptotically efficient even under agents' strategic behaviors; this largeness concept captures the distributed production economies that we study.\par

\cite{LLSB08} were the first to model peer-to-peer networks as distributed markets. Along with 
\cite{WZ2007} who showed convergence of the PR algorithm to equilibria in these models, they stimulated a sequence of works, already discussed above, on the stability of PR algorithms in various applications. Recently, the study of distributed production economies regained traction in the context of the emerging cryptocurrency markets \cite{Chen19}. However, although critical for their stability, the incentives of miners to allocate resources among multiple markets are still not well understood, \cite{Bis19,Fia19,Gor19}.

Another line of research focuses on explaining the growth of (production) economies in the long run, rather than equilibration in the short run. Recently, \cite{BMN2018} proposed a dynamical variant of von Neumann's pioneering model on economic growth \cite{VonNeumann1971}. They showed that the use of PR algorithm to exchange produced goods (which are then used as resources for future productions) leads to universal growth of economies under mild conditions on the efficiencies of the firms.

Finally, various forms of chaos in market dynamics are receiving increasingly more attention ~\cite{SAF2002,GF2013,Pal17,CheungP19,CheungP20,CT21,CFMP2020,Bie21,Cho21,Leo21}.

\paragraph{Paper Outline.}

Section~\ref{sec:model} presents our three models: Cournot competition with multiple-goods, Tullock contests and Fisher Markets. We discuss their mathematical connections. Section~\ref{sect:results} presents our main results: convergence of PR dynamics and chaos and instabilities of GA and BR dynamics. 
Detailed proofs are delegated to the appendix, but in Sections \ref{sect:PR} and \ref{sect:liyorke} we discuss the techniques we use.
\section{Models and Definitions}\label{sec:model}

\newcommand{\bbxs}{\bbb^*}
\newcommand{\bbps}{\bbp^\#}

In this section, we describe the Cournot competition and Fisher market models. In their classical descriptions,
\emph{quantities of goods produced} are used as the driving variables to define the notions of Nash and market equilibria.
However, it will be more convenient to use \emph{spendings/investments on the production of a good} as the driving variables here,
since this is the domain of the PR algorithm.
In all models, $N = \{1,2,\cdots,n\}$ is the set of firms (agents) and $M = \{1,2,\cdots,m\}$ is the set of goods.

\paragraph{Multi-good Cournot Competition (CC) with Isoelastic Demands.}
Each firm $i$ invests an amount $b_{ij}\ge 0$ on producing good $j$. We write $\bbb_i := (b_{ij})_{j\in M}$ and $\bbb:=\(\bbb_i\)_{i\in N}$.
Each firm $i$ has only finite amount of capital, $K_i$, to invest, thus it is subject to a capital constraint $\sum_j b_{ij} \le K_i$.
We assume that the marginal cost of producing good $j$ is the same for all firms, which we denote by $\alpha_j$. Thus, the quantity of good $j$ produced by firm $i$ is $b_{ij}/\alpha_j$. Each good $j$ has isoelastic demand, i.e., the total sales revenue of the good is constant, denoted by $v_j$. Thus, the price function\footnote{
We also consider more general price functions induced by different demand elasticities in \Cref{sect:liyorke}.} for good $j$ is $P_j(\bbb) ~:=~ v_j/\left(\sum_i b_{ij}/\alpha_j\right)$,
and the revenue of firm $i$ received from the sales of good $j$ is
$P_j(\bbb) \cdot (b_{ij}/\alpha_j):= v_j \cdot y_{ij}$, where 
$y_{ij}$ denotes the \emph{market share} of firm $i$ on good $j$:
\begin{equation}\label{eq:csf}
y_{ij} := b_{ij}/\sum_k b_{kj}.
\end{equation}
The profit of firm $i$ is its revenue from the sales of all goods minus its total investment: $\sum_j v_j y_{ij} - \sum_j b_{ij}$.

\paragraph{Tullock Contest (TC).} The above setting admits a correspondence to \emph{multiple Tullock contests}. According to this interpretation, each firm $i$ invests an amount of $b_{ij}\ge 0$ on producing good $j$,
but now the goods are considered as \emph{prizes}, and the probability that firm $i$ wins good $j$ is $y_{ij}$ as defined in~\eqref{eq:csf}.
This probabilistic interpretation is natural in the applications of e.g., blockchain mining and imperfectly discriminating all-pay auctions (crowdsourcing). Now, different firms can have different valuations on the prize, so the parameter $v_j$ in CC may be distinct for different firms;
we let $\v$ denote the valuation of firm $i$ on good $j$. The expected profit of firm $i$ is
\begin{equation}\label{eq:utility_cc}
u_i(\bbb_i):=\sum_j v_{ij} y_{ij} - \sum_j b_{ij}.
\end{equation}
While CC and TC have differences in their rationales, they admit a correspondence in mathematical terms, by replacing deterministic profit in CC with expected profit in TC, and $v_j$ with $\v$ for different firms $i$. Accordingly, we will henceforth refer to this model as CC/TC or simply TC.

\begin{definition}[Nash equillibrium]
For any $\delta\ge 0$, we say that $\bbxs $ is a \emph{$\delta$-Nash Equilibrium} ($\delta$-NE) of a CC/TC if for each agent $i\in N$,
\[\max_{\bbb_i: \sum_j \b \le K_i} u_i(\bbb_i,\bbxs_{-i}) \le (1+\delta)\cdot u_i(\bbxs_i,\bbxs_{-i}).\]
In other words, agent $i$ cannot improve her utility by more than an $\delta$ fraction at $\bbxs$ by unilaterally changing her own investment portfolio. We call a $0$-NE simply a NE.
\end{definition}

\paragraph{Fisher Market (FM).} In a Fisher market, each good $j$ has a supply which is normalized to one unit. Again, $b_{ij}$ denotes the spending of firm $i$ on good $j$, and each firm $i$ has a budget of $K_i$, so the constraint $\sum_j b_{ij} \le K_i$ applies. Let $\bbp=\(p_j\)_{j\in M}$, where $p_j$ denotes the price of good $j$. 
At $\bbb_i$, firm $i$ gets $b_{ij}/p_j$ units of good $j$ and has a \emph{quasi-linear} utility function,
$u_i\(\bbb_i\mid \bbp\)$, which takes the form
\begin{equation}\label{eq:utility_fm}
u_i(\bbb_i\mid \bbp) = \sum_j \v \cdot (b_{ij}/p_j) - \sum_j b_{ij},
\end{equation}
where $v_{ij}$ denotes firm $i$'s valuation of one unit of good $j$. At price vector $\bbp$, each firm $i$ select an \emph{optimal budget allocation} $\bbb_i^\#$ in $\argmax_{\bbb_i} u_i(\bbb_i \mid \bbp )$ which maximizes its utility subject to the constraint $\sum_j b_{ij} \le K_i$. At an optimal budget vector $\bbb^\#_i$, a vector $\bbx^\#_i:=(b_{ij}^\#/p_j)_{j\in M}$ is called a \emph{production bundle} of agent $i$ at price vector $\bbp$. 

\begin{definition}[Market equilibrium]
A price vector $\bbps = (p^\#_j)_{j\in M}$ is a \emph{market equilibrium (ME)} if there exists an optimal budget allocation $\bbb^\#=(\bbb_i^\#)_{i\in N}$ at $\bbps$, such that for each good $j$, $\sum_i b^\#_{ij} = p^\#_j$. The vector $\bbb^\#$ is called a \emph{market equilibrium spending}.\footnote{The last condition is same as $\sum_i b^\#_{ij}/p_j^\# = 1$, which is the classical definition of market equilibrium.}
\end{definition}

\paragraph{Connection between TC and FM.}
The crucial difference between TC and FM is that in TC, prices are determined endogenously as a function of $\bbb$, whereas in FM, prices are viewed as independent inputs that do not explicitly depend on $\bbb$. Thus, while both models require each firm $i$ to make an allocation $\bbb_i$ that is subject to the same budget constraint $\sum_j b_{ij} \le K_i$, the methods to determine outcomes differ.

However, if $\(\bbp,\bbb\)$ are market equilibrium and market equilibrium spending respectively of an FM, then $\sum_i b_{ij}/p_j = 1$ for each good $j$. Thus, we can translate $b_{ij}/p_j$, which is the quantity of good $j$ that firm $i$ gets at the market equilibrium, to the probability that firm $i$ wins good $j$ in the corresponding TC. Under this translation, the outcome in the FM is the same as the outcome in the TC. Due to the well-known properties of Fisher markets, this outcome is Pareto-optimal, and it is envy-free if $K_i$ is identical for all $i$.

The above suggest that if there is an algorithm that converges to the market equilibrium spending (our \Cref{thm:pr-qlin} establishes this) of the FM, then it yields a feasible solution of the corresponding TC. The remaining question is the quality of this feasible solution, i.e., how close it is to a Nash equilibrium of the TC. It turns out that if the underlying distributed production economy satisfies a natural \emph{largeness} property, then the market equilibrium spending is also a $\delta'$-NE for some small $\delta'>0$. In particular, as we show in \Cref{pr:me-to-ane} below, this is the case if the budget of each firm is small compared to \emph{any} market equilibrium price, i.e., if $\max_{i,j} \{K_i/p^\#_j \}\le \delta$ for a small $\delta>0$.
We may view $\delta$ as a parameter that describes the largeness of the economy: the smaller $\delta$ is, the \emph{larger} the economy is.
(We also need the \emph{bang-per-buck ratio} $\beta_i := \max_j \{v_{ij}/p^\#_j\}$ to be sufficiently high for all firms $i$,
because otherwise a firm might invest nothing 
thus attain zero utility, forcing $\delta'$ to be $+\infty$.)

\begin{proposition}\label{pr:me-to-ane}
Suppose that $\bbb^\#$ is a market equilibrium spending vector of a quasi-linear FM, and $\bbp^\#$ is the corresponding market equilibrium price vector.
For every $i\in N$, let $\beta_i := \max_j \{v_{ij}/p^\#_j\}$. 
If $\max_{i,j} \{K_i/p^\#_j \}\le \delta$, then $\bbb^\#$ is also a $\delta'$-NE of the corresponding TC, where
$\delta' = \max_{i: \beta_i > 1} \left\{\(\frac{\beta_i}{1-\delta} - 1\)/\(\beta_i-1\)\right\}- 1$,
provided that there is no firm $k$ with $1-\delta < \beta_k < 1$. 

\end{proposition} \noindent
It is easy to see that if $\min_i \{\beta_i\}$ grows, then $\delta'$ tends toward $\delta/(1-\delta)$.
The proof of \Cref{pr:me-to-ane} can be found in \Cref{app:equilibria}.
\section{Our Main Results}\label{sect:results}

We present our two main results here. We discuss the methodology of proving them in Sections~\ref{sect:PR} and \ref{sect:liyorke}.

\paragraph{Proportional Response (PR) in Quasi-linear Fisher Market.} In a quasi-linear Fisher market,
our PR protocol starts with each firm $i\in N$ investing an arbitrary portfolio $\bbb_i^\circ$ which is \emph{positive},
i.e., $b_{ij}^\circ>0$ for all $j\in M$.
In each round, firms update their portfolios simultaneously according to the \textsf{PR-QLIN} protocol in Algorithm~\ref{alg:pr_qlin}.

The \textsf{PR-QLIN} protocol can be naturally interpreted.
After all firms update their investment portfolios in round $t$, one unit of each good $j$ is allocated to the firms
in proportion to their investments on the good. Thus, firm $i$ gets $y_{ij}^t$ units of good $j$ (line 3).
Then each firm $i$ computes its attained utility, $S_i^t$, without subtracting investment cost (line 5).
If $S_i^t > K_i$, then firm $i$ will appropriate all of its capital, $K_i$, for investment in round $t+1$;
otherwise it will only appropriate an amount of $S_i^t$ for investment. 
Then each firm invests its appropriated capital on each good in proportion to the utility attained from that good in the previous round,
i.e., firm $i$ invests a fraction of $\v y_{ij}^t/S_i^t$ of its appropriated capital on good $j$.
Our main result is stated below.

\begin{theorem}\label{thm:pr-qlin}
Given any positive starting point $\bbb^\circ$, the algorithm \emph{\textsf{PR-QLIN}} converges to the set of market equilibrium spending vectors
of the quasi-linear Fisher market.
\end{theorem}

\begin{algorithm}[!htb]
\caption{PR-QLIN Learning Protocol}\label{alg:pr_qlin}
\textbf{Input:} $\(K_i,v_{i1},v_{i2},\ldots,v_{im},\bbb_i^\circ\)$ for each firm $i$\\
\textbf{Output:} market equilibrium spending $\bbb^\#$.\\[-0.3cm]
\begin{algorithmic}[1]
\For{$t=0,1,2,\ldots$}
\For{every firm $i$ and good $j$}
\State{$y_{ij}^t\gets \b^t/\sum_{k}b_{kj}^t$}
\EndFor
\For{every firm $i$}
\State{$S_i^t\gets \sum_{j}\v y_{ij}^t$} 
\For{every good $j$}
\If {$S_i^t>K_i$}
\State {$\b^{t+1}\gets (\v y_{ij}^t / S_i^t)\cdot K_i$}
\Else \State {$\b^{t+1}\gets \v y_{ij}^t$} $\Comment{\text{same as }(\v y_{ij}^t / S_i^t)\cdot S_i^t}$
\EndIf
\EndFor
\EndFor
\EndFor
\end{algorithmic}
\end{algorithm} 

\paragraph{Gradient Ascent Dynamics and Li-Yorke Chaos.} To establish our chaos results of the GA dynamics in CC (hence, also in TC),
we consider a CC with one good and $n$ firms. Since there is only one good, we omit the subscript $j=1$ and use the shorthand $\alpha \equiv \alpha_1$ to denote the marginal cost of producing the good (recall from \Cref{sec:model} that this is equal for all firms). In this setting,
it is more convenient to use the quantities of the good produced, i.e., the variables $x_i = b_{i1} / \alpha$, as the driving variables.
Without loss of generality, let $v_1 = 1$.
Then the utility of firm $i$ is $u_i(\bbx) = x_i / (\sum_k x_k) - \alpha x_i$.
The \emph{Gradient Ascent (GA)} update rule is given by
$x_i^{t+1} \la x_i^t+\eta \cdot \nabla_{i}u_i\(\bbx^t\)$, where $\eta$ is the step-size.\par
Assuming that the initial point is \emph{symmetric}, i.e., 
that $x_i^\circ$ is identical for all $i$, then, in each round $t>0$, the $x_i^t$'s remain identical for all $i$.
Thus, a symmetric GA dynamic is essentially one-dimensional, and its trajectory can be represented by the sequence $\{x_1^t\}_{t\ge 0}$
generated by the GA update rule:
\begin{equation}\label{eq:grad_ascent}
\textstyle x_1^{t+1} \la x_1^t+\eta \cdot \(\frac{n-1}{n^2 x_1^t}-\alpha\) .
\end{equation}
Our main result states that even for such an apparently simple one-dimensional dynamical system, \emph{chaos} occurs with step-size $\eta$ as small as $\Theta(1/n)$.
Here, we refer to \emph{Li-Yorke chaos} which is formally defined below.
\begin{definition}[Li-Yorke Chaos]\label{def:liyorke}
A discrete time dynamical system $\(x^t\)_{t\in \mathbb N}$ such that $x^{t}:=f^t\(x^\circ\)$ for a continuous update rule $f: X \to X$ on a compact set $X\subseteq \mathbb R$ is called \emph{Li-Yorke chaotic}, if (i) for each $k \in \mathbb N$, there exists a periodic point $\hat x \in X$ with period $k$, and (ii) there is an uncountably infinite set $S\subset X$ that is \emph{scrambled}, i.e., if for each $x\neq x' \in S$ it holds that $\lim\inf_{t\to \infty} |f^t\(x\)-f^t\(x'\)|=0<\lim\sup_{t\to\infty}|f^t\(x\)-f^t\(x'\)|$.
\end{definition}

\begin{theorem}[Li-Yorke Chaos in $n$-Player CC/TC]\label{thm:liyorke}
Consider a symmetric GA dynamic with $n$ firms and marginal cost $\alpha > 0$.
Then for any step-size $\eta \ge 3(n-1)/n^2\alpha^2$, the essentially-one-dimensional dynamical system \eqref{eq:grad_ascent} is Li-Yorke chaotic.
\end{theorem}

This theorem applies with isoelastic price function. In Section \ref{sect:liyorke}, 
we consider a larger family of price functions and prove that Li-Yorke chaos also occurs in the corresponding symmetric GA dynamics. We also present theoretical and empirical evidences that instability arises when the GA rule is replaced by the Best Response rule.
\begin{remark}
In practice, firms may choose to use a large step-size in a myopic, greedy approach to profit maximization.
Given that chaos occurs with a vanishingly small step-size $\Theta\(1/n\)$ as the number of firms increases (cf. \Cref{thm:liyorke}),
our result is practically relevant for distributed production economies in which many small firms are involved.
Stability results should be possible for smaller step sizes, however, such step sizes are not particularly interesting from a practical perspective. Finally, the presence of a centralised planner who may enforce small step sizes is a rather unnatural assumption for the settings and applications that we consider. 
\end{remark}
\section{Proportional Response Dynamics}\label{sect:PR}

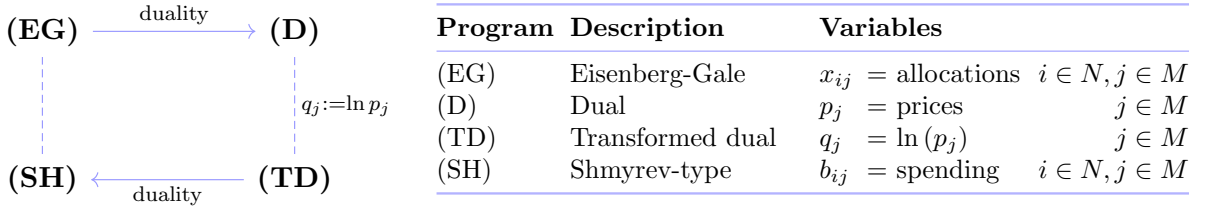
\begin{figure*}[!htb]
\tikzcdset{row sep/normal=1.3cm}
\tikzcdset{column sep/normal=2cm}
\begin{minipage}[t]{0.34\textwidth}
\centering
\begin{tikzcd}
\textbf{(EG)} \arrow[r, blue!40, "\text{duality}" black]
& \textbf{(D)} \arrow[d, dashed, no head, blue!40, "q_j:=\ln{p_j}" black] \\
\textbf{(SH)} \arrow[u, blue!40, dashed, no head]
& \textbf{(TD)} \arrow[l, blue!40,"\text{duality}" black]
\end{tikzcd}
\end{minipage}\hspace{7pt}
\begin{minipage}[t]{0.64\textwidth}
\centering
\arrayrulecolor{blue!30}
\setlength{\tabcolsep}{8pt}
\small
\begin{tabular}{@{}l@{\hspace{6pt}}ll@{}c@{}l@{}r@{}}
\toprule
\textbf{Program} & \textbf{Description} &\multicolumn{4}{l}{\textbf{Variables}}\\
\midrule
(EG) & Eisenberg-Gale & $\x $ &\, $ =$\,\, & allocations \,\,& $i\in N, j\in M$\\
(D) & Dual & $p_j$ &\, $ =$\,\,& prices & $j\in M$\\
(TD) & Transformed dual & $q_j$ &\, $ =$\,\, & $\ln{\(p_j\)}$ & $j\in M$\\
(SH) & Shmyrev-type & $\b$ &\, $  =$\,\, &spending & $i\in N, j\in M$\\
\bottomrule
\end{tabular}
\end{minipage}
\caption{Convex programs in the derivation of the \textsf{PR-QLIN} protocol via the Mirror Descent (MD) protocol. Starting from the dual (D) of a generalized Eisenberg-Gale convex program (EG), we go to the transformed dual (TD) and by convex duality to a Shmyrev-type primal program (SH) which is hence, equivalent to the initial program (EG). The objective functions of (SH) for quasi-linear utilities is 1-Bregman convex which implies convergence of the MD protocol. }
\label{fig:derivation}\vspace*{-0.1cm}
\end{figure*}

Our proof of Theorem~\ref{thm:pr-qlin} consists of two major steps. In the first, we derive a convex program that captures the market equilibrium (ME) spending of the quasi-linear Fisher market via the approach of \cite{Bir11,Col17,Che18}. In the second, we show that a general Mirror Descent (MD) algorithm converges to the optimal solution of this convex program; \textsf{PR-QLIN} is an instantiation of this MD algorithm.

\paragraph{Convex Program Framework.} 
We first utilize a convex optimization framework to derive a convex program that captures the ME spendings of any quasi-linear FM.
The ensuing framework is summarized in Figure~\ref{fig:derivation}. In short, via duality and variable transformations, the market equilibria of a FM can be captured by various convex programs, each with a different domain.\footnote{For linear Fisher markets, i.e. markets in which each agent has a utility similar to a quasi-linear utility, but without the subtraction of investment cost, \cite{EG59} derived a convex program which captures the ME allocation, where the driving variables are quantities of goods allocated to the agents. Subsequent works 
established that by considering suitable duals and transformations of Eisenberg and Gale's convex program, new convex programs can be derived which capture the ME prices and ME spendings.}
Our starting point is a convex program proposed by \cite{Dev09} that captures ME prices of quasi-linear Fisher market (which belongs to type (D) in Figure~\ref{fig:derivation}). From this, we derive a new convex program with captures the ME spendings of the market (which belongs to type (SH));
see \Cref{app:quasi-linear-ces} for the details. The convex program is
\begin{alignat*}{3}
\min_{\bbb,\bbw,\bbp} &~~-\sum_{i=1}^n\sum_{j=1}^m\b\ln{v_{ij}} +\sum_{i=1}^n w_i +\sum_{j=1}^m p_j\ln{p_j} \\
\text{s.t.}\,\,& \sum_{i=1}^n \b = p_j,\hspace{38pt} \forall j\in M,\\
& \sum_{j=1}^m \b+w_i= K_i,~~~ \forall i\in N,\tag{SH}\\
& \b,w_i\ge 0, \hspace{47pt}\forall i\in N, j\in M.
\end{alignat*}
For convenience, we will write $F(\bbb,\bbw,\bbp):=-\sum_{i=1}^n\sum_{j=1}^m\b\ln{v_{ij}} +\sum_{i=1}^n w_i +\sum_{j=1}^m p_j\ln{p_j}$.
Observe that the first and second constraints determine the values of $\bbw,\bbp$ in terms of $\b$'s.
Thus, we can rewrite the convex program to have variables $\bbb$ only,
and the remaining constraints are $\b\ge 0$ and $\sum_{j=1}^m \b \le K_i$;
we slightly abuse notation by using $F(\bbb)$ to denote the objective of this convex program.

\paragraph{From Mirror Descent to Proportional Response.} After having the convex program with variables $\bbb$ only,
we can compute a ME spending by the optimization algorithm of Mirror Descent (MD).
To begin, we recap a general result about MD~\cite{CT93,Bir11}.

\begin{definition}[KL-divergence and $L$-Bregman convexity]
Let $C$ be a compact and convex set and let $h$ be a convex function on $C$. Then, for any $\bbz'\in C, \bbz \in \rint(C)$ where $\rint(C)$ is the relative interior of $C$, the \emph{Bregman divergence}, $d_h\(\bbz',\bbz\)$, generated by $h$ is defined by
\[d_h(\bbz',\bbz) := h(\bbz') - \left[h(\bbz) +\inner{\nabla h(\bbz)}{\bbz'-\bbz}\right].\]
The \emph{Kullback-Leibler (KL) divergence} between $\bbz'$ and $\bbz$ is defined by
$\KL(\bbz' \| \bbz) := \sum_j z'_j \cdot \ln \frac{z'_j}{z_j} - \sum_j z'_j + \sum_j z_j$,
which is the same as the Bregman divergence $d_h$ with regularizer $h\(\bbz\) := \sum_j (z_j \cdot \ln z_j - z_j)$. A function $F$ is \emph{$L$-Bregman convex} w.r.t. the Bregman divergence $d_h$ if
for any $\bbz'\in C$ and $\bbz\in \rint(C)$, $F(\bbz) + \inner{\nabla F(\bbz)}{\bbz'-\bbz}\le f(\bbz') \le f(\bbz) + \inner{\nabla f(\bbz)}{\bbz'-\bbz} + L \cdot d_h(\bbz',\bbz)$.
\end{definition}

For the problem of minimizing a convex function $F(\bbz)$ subject to $\bbz\in C$,
the MD protocol w.r.t.~the KL divergence is presented in Algorithm~\ref{alg:md}.
In the protocol, $1/\Gamma$ is the step-size, which may vary with $t$ (and typically diminishes with $t$). However, in the current application of distributed dynamics, a time-varying step-size 
is undesirable or even impracticable, since it requires firms to keep track of a global clock.


\begin{algorithm}[!htb]
\caption{MD protocol w.r.t. $\KL$-divergence}\label{alg:md}
\textbf{Input:} A convex set $C$, a function $F$ defined on $C$, a parameter $\Gamma$ and a point $\bbz^\circ \in C$.\\
\textbf{Output:} $\bbz^*=\argmin_{\bbz \in C} F(z)$. 
\begin{algorithmic}[1]
\For{$t=0,1,2,\ldots$}
\State{$g\(\bbz,\bbz^t\)\gets \inner{\nabla F(\bbz^t)}{\bbz-\bbz^t}+\Gamma\cdot \KL(\bbz\|\bbz^t)$}
\State{$\bbz^{t+1}\la \argmin_{\bbz\in C}\{g\(\bbz,\bbz^t\)\}$}
\EndFor
\end{algorithmic}
\end{algorithm}

\begin{theorem}
\label{thm::plain::convex}
Suppose $F$ is an $L$-Bregman convex function w.r.t. the Bregman divergence $d_h$,
and $\bbz^t$ is the point reached after $t$ applications of the MD update rule in Algorithm~\ref{alg:md} with parameter $\Gamma = L$. Then
$F(\bbz^t) -F(\bbz^*) ~\le~ L \cdot d(\bbz^*, \bbz^\circ)/t,$
where $\bbz^*=\argmin_{\bbz \in C} F(z)$. 
\end{theorem}
\noindent In \Cref{app:PR}, we first prove Lemma~\ref{lem:one-Bregman-convex} below. Then, we show that \textsf{PR-QLIN} is an instantiation of Algorithm~\ref{alg:md} with $\Gamma=1$. This is achieved by identifying the variables $\bbb$ in \textsf{PR-QLIN} as the variables $\bbz$ in Algorithm~\ref{alg:md},
and the domain $\{\bbb ~|~\b \ge 0~\text{and}~\sum_{j=1}^m \b \le K_i\}$ as the convex set $C$ in Algorithm~\ref{alg:md}. Thus, Theorem~\ref{thm::plain::convex} guarantees the updates of \textsf{PR-QLIN} converge to an optimal solution of the convex program~\eqref{eq:shmyrev-quasilinear}, and hence Theorem~\ref{thm:pr-qlin} follows.

\begin{lemma}\label{lem:one-Bregman-convex}
The objective function $F(\bbb)$ of \eqref{eq:shmyrev-quasilinear} is a $1$-Bregman convex function w.r.t. the KL-divergence.
\end{lemma}
\section{Gradient Ascent \& Best Response Dynamics}\label{sect:liyorke}

To establish the statement of Theorem~\ref{thm:liyorke} about the GA dynamics in equation \eqref{eq:grad_ascent} for $n=2$ (the technique is similar for any $n>2$), let $f\(x\):=x+\eta\(\frac1{4x}-\alpha\)$. To prove that Li-Yorke chaos occurs, we use a seminal theorem of \cite{Li75}, which states that if $f$ has two easy-to-verify properties, then the dynamical system is Li-Yorke chaotic.
The two properties are: (i) an invariant set of $f$ that includes a fixed point $x^*$, i.e., an interval $I=[L,U]$ such that $f\(I\)\subseteq I$ with a point $L<x^*<U$ satisfying $f(x^*) = x^*$,
and (ii) a point $x'\in I$ other than $x^*$ with period $3$, i.e., $f^{\(3\)}\(x'\)=x'$, where $\fof:=\(f\circ f\circ f\)\(x\)$. These properties are formally established in Lemma \ref{lem:invariant} and \Cref{lem:periodic_3} respectively. A visualization of Theorem~\ref{thm:liyorke} is provided in the first two panels of Figure \ref{fig:liyorke}. It can be seen that chaos may emerge even for small step-size and for asymmetric marginal costs. 

\begin{figure*}[!t]
\centering
\includegraphics[width=0.49\linewidth]{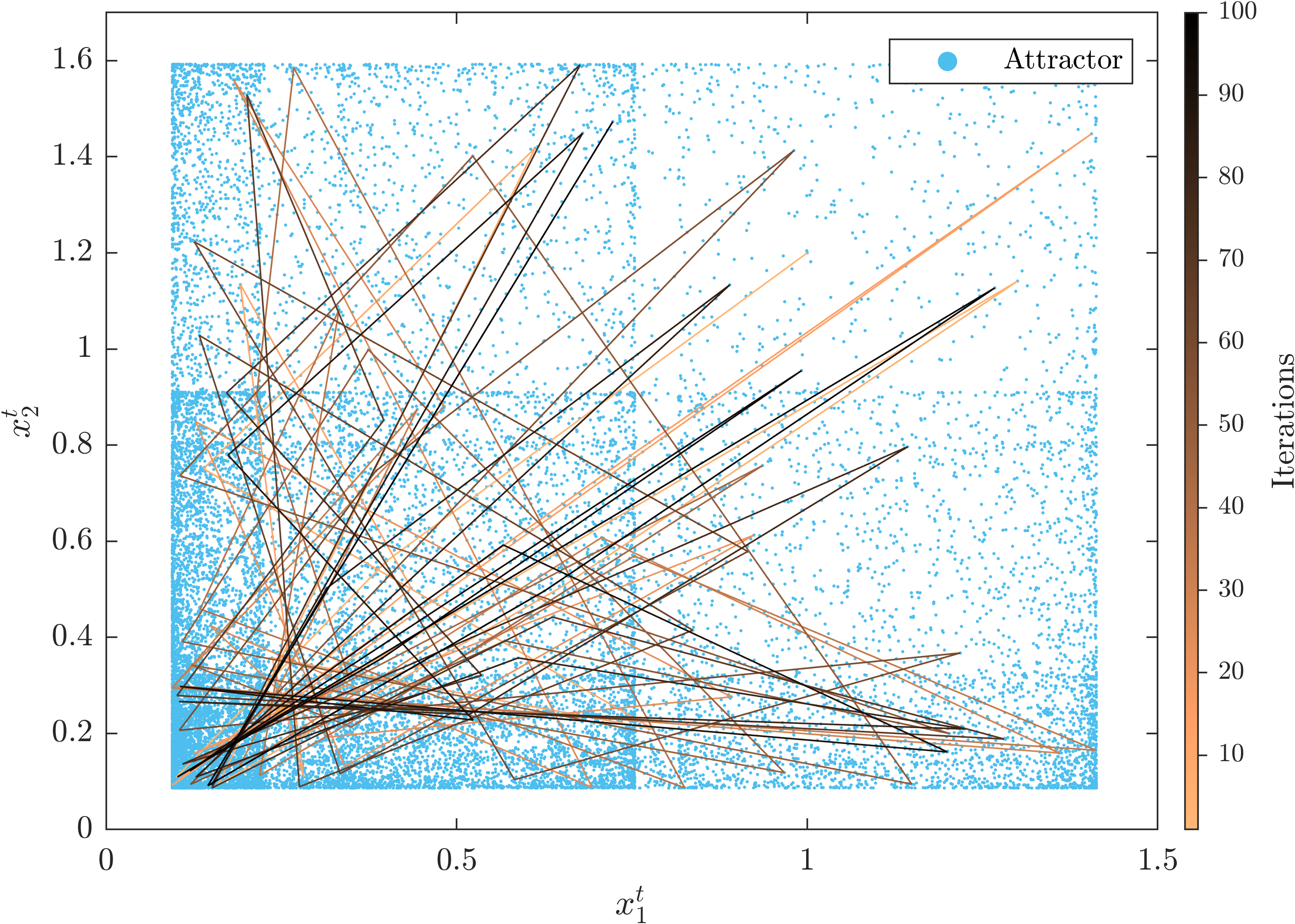}\hspace{0.2cm}
\includegraphics[width=0.49\linewidth]{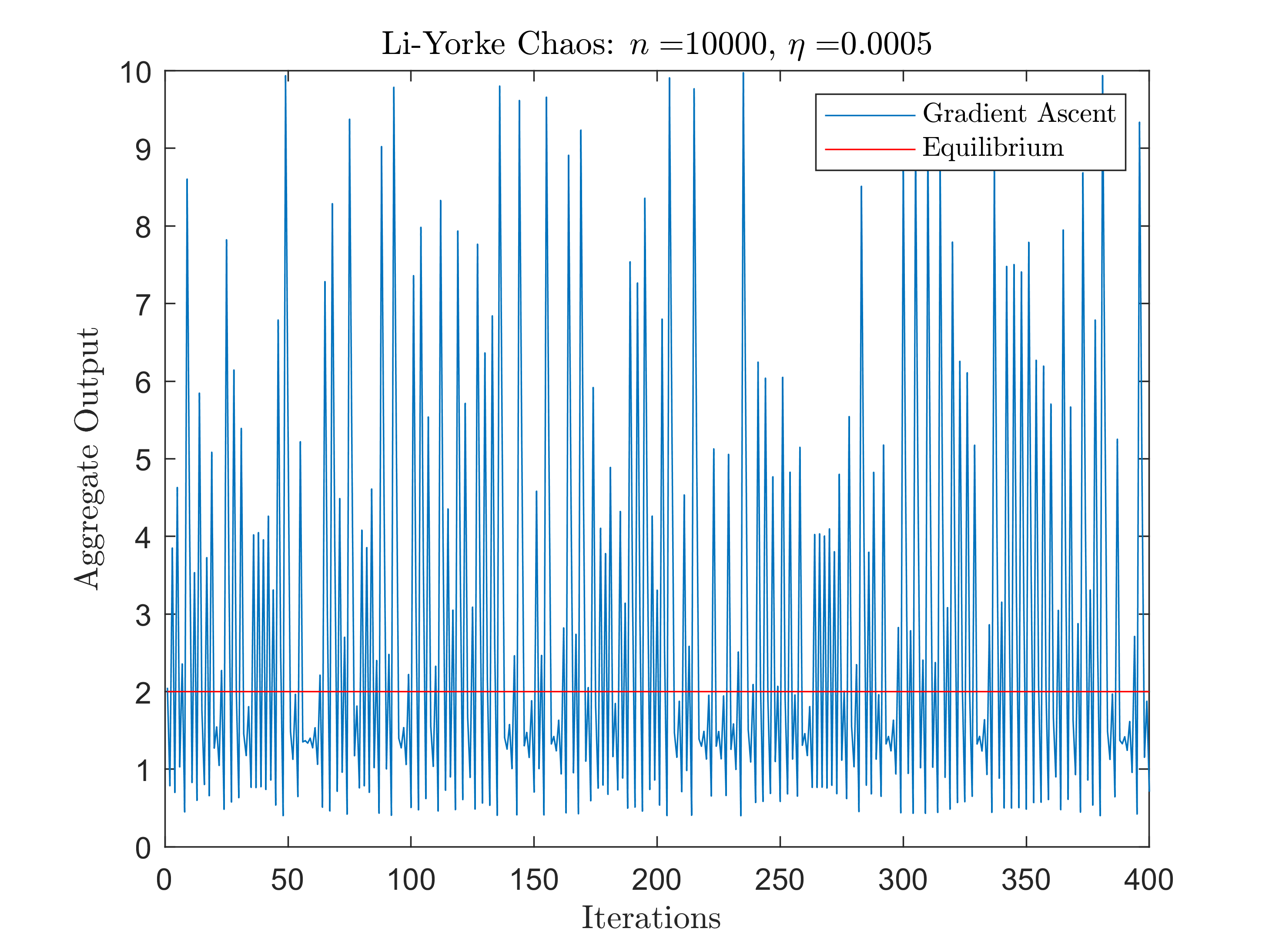}\\
\includegraphics[width=0.48\linewidth]{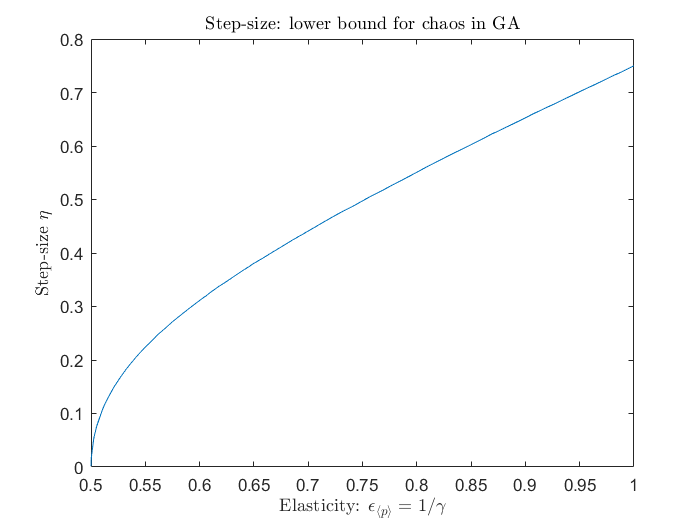}
\caption{Li-Yorke chaos of the Gradient Ascent (GA) dynamics with constant step-size in $n$-firm Cournot competition with isoelastic inverse demand function (equivalently, Tullock contest with proportional success function). First panel: chaotic trajectories (light to dark lines) of the output pairs of two firms that start from different initial outputs for $t\in [1,200]$. The planar dots denote their projections for $t\in[1,10^4]$ and fill the plane as $t$ grows. Second panel: aggregate output in a market with $n=10^4$ firms with different costs $a_i\in[10^{-5}, 1], i=1,\dots,10^4$ starting from a randomly selected initial output vector and using GA with step-size $\eta=5\cdot 10^{-4}$ for $t\in[0,400]$. Third panel: Minimum step-size for which chaos provably occurs under the GA update rule in 2-firm Cournot competition with inverse demand function $\(x_1+x_2\)^{-\gamma}$, for $\gamma>0$. Interestingly, chaotic behavior is more likely when demand is inelastic.} 
\label{fig:liyorke}
\end{figure*}

\begin{figure}[!tb]
\centering
\includegraphics[scale=0.9]{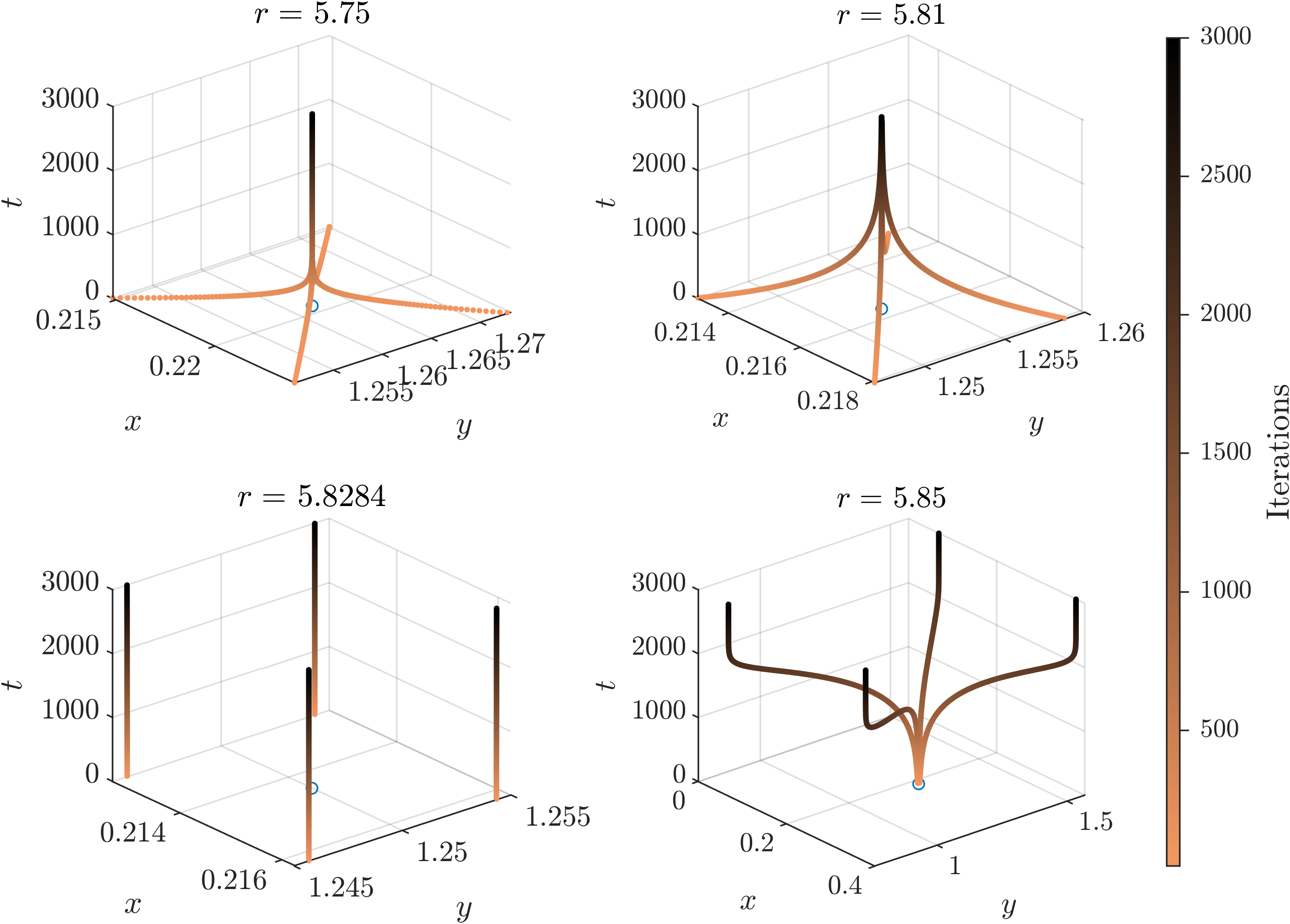}
\caption{Firms' outputs $\(x^t,y^t\)$ (horizontal plane) with respect to time $t\in[10,250]$ (vertical axis). Parameter $r$ captures the cost asymmetry between the two firms. The dynamics spiral inwards towards the equilibrium (red circle on the $\(x,y\)$-plane) for lower asymmetry between the two firms (upper panels), cycle around the equilibrium when $r$ is equal to the critical level $r_0\approx 5.8284$ (bottom left panel) and spiral outwards (till they become cyclic) for values of $r>r_0$ (bottom right panel). This is in agreement with Proposition \ref{prop:stability} in \Cref{app:liyorke} which shows that the attractors of the BR dynamics may change significantly for only small perturbations in the system parameters as expressed by the cost asymmetry-ratio $r$.}
\label{fig:spiral_1}
\end{figure}

\paragraph{General Price Functions.} 
The emergence of chaos is not tied to the particular choice of the isoelastic price function of the form $v_1 / X$,
where $X := \sum_i x_i$. More generally, we may consider the parametric price function $X^{-\gamma}$,
where $\gamma>0$ is the inverse of the demand elasticity $\varepsilon_{\langle p \rangle}$. 
This is a special case of the more general price functions $P\(X\)=A+BX^{-\gamma}$, see e.g., \cite{Lop19}, for $A=0,B=1$ and $\gamma>0$. For these functions, we verify the two conditions required in the theorem of Li and Yorke, for a range of choices of $\gamma$ numerically (via computer software).
The lower bound of the step-size $\eta$ at which chaos emerges (in the symmetric case) is depicted in the third panel of Figure \ref{fig:liyorke}.
The interesting observation is that chaos is more likely for less elastic demand.

\paragraph{Best Response Dynamics.} We conclude by revisiting the well-studied Best Response (BR) dynamics and formally establish that they can be unstable even in the simplest setting of two firms and one good. The general BR update rule is $x_i^{t+1} \la \argmax_{x_i} u_i(x_i,x_{-i}^t)$. For TC with isoelastic demand, the BR dynamics take the form $x_i^{t+1}\la\sqrt{x_{-i}^t/\alpha_i}-x_{-i}^t,$ for $i=1,2$, where $\alpha_i$ is the marginal cost of firm $i$. 
BR dynamics in Cournot duopoly with isoelastic functions have been (empirically) studied by \cite{Puu91} and, in the framework of contests, in \cite{War18}. Both papers suggest that the stability of the unique fixed point, $\(x_1^*,x_2^*\)$, depends on the degree of asymmetry between the two firms, captured by the ratio $r:=\alpha_1/\alpha_2$ with instabilities emerging as the asymmetry increases. While our previous technique does not lead to a formal proof of Li-Yorke chaos in  BR dynamics, we formalize the (in)-stability properties of the latter via eigenvalue analysis of a first-order linear approximation of the original non-linear system, cf. Proposition \ref{prop:stability}. The result is visualized in Figure \ref{fig:spiral_1} which shows how the trajectories of the dynamics may change dramatically in response to even small perturbations of the model parameters (firms' costs).


\section{Conclusions}\label{sec:conclusions}
The current work brings together multi-agent learning with optimization, market theory and chaos theory. Our findings suggest that by considering production economies from a market rather than a game-theoretic perspective, we can formally derive a natural learning protocol (PR) which is stable and converges to effective outcomes rather than being chaotic (GA). Due to its simple form and mild informational requirements, PR can be used to study real-world multi-agent settings from an AI perspective. Since distributed production economies capture many important applications (blockchain, peer-to-peer networks, crowdsourcing), our contributions are significant both for theoretical and practical purposes.


\section*{Acknowledgements}
Yun Kuen Cheung and Stefanos Leonardos gratefully acknowledge NRF 2018 Fellowship NRF-NRFF2018-07.
Georgios Piliouras gratefully acknowledges grant PIE-SGP-AI-2020-01, NRF 2019-NRF-ANR095 AL\-IAS grant and NRF 2018 Fellowship NRF-NRFF2018-07.

\bibliographystyle{plain}
\bibliography{tullock_bib}

\newpage
\appendix
\section{Convex Programs and Equilibrium Characterizations}

In this appendix, we prove \Cref{pr:me-to-ane}. To do so, our first task is to characterise the Nash equilibria of the multi-good Cournot competition (CC) and simultaneous Tullock contests (TC) models and compare them to the Market equilibria of the ensuing Fisher Market with quasi-linear utilities. This can be conveniently done by expressing all relevant quantities via proper convex optimization problems and leveraging techniques from convex duality. These are done in \Cref{app:equilibria,app:quasi-linear-ces}. The proof of \Cref{pr:me-to-ane} is presented in \Cref{app:me-to-ane}.

\subsection{Nash Equilibria in Cournot Competition and Tullock Contests}\label{app:equilibria}

To proceed, let $\Gamma=\(N,M,\(u_i\)_{i\in N}\)$ denote a CC or TC model and recall that a Nash equilibrium $\bbb^*=\(\bbb^*_i\)_{i\in N}$ of  is a strategy profile such that the strategy $\bbb_i^*=\(b^*_{i,j}\)_{j\in M}$ of each firm $i \in N$ maximizes its own utility given the strategies $\bbb^*_{-i}$ of the other players. Setting $p_j:=\sum_{i=1}^n \b$, for each $j\in M$, the equilibrium strategy $\bbb_i^*$ can be expressed as the solution 
of the following maximization problem
\begin{alignat*}{3}
\max_{\bbb_i}\,\,&  u_i\(\bbb_i;\bbb^*_{-i}\) = \sum_{j=1}^m \b\(\frac{\v}{p_j}-1\)  \\
\text{s.t.}\,\,& \sum_{j=1}^m \b \le K_i,\tag{NE}\\
& \b  \geq 0,  \qquad \quad \forall j\in M.
\end{alignat*}
The following characterization stems from the first-order conditions.
\begin{proposition}\label{prop:foc}
At any Nash Equilibrium, $\bbb^*=\(\bbb^*_i\)_{i\in N}$, of $\Gamma$, it holds that
\begin{enumerate}[label=(\roman*), align=left]
\item If firm $i$ exhausts its capacity, i.e., if $\sum_{j=1}^{m} b^*_{ij}=K_i$, then there exists a constant $C_i^*\ge 1$ such that 
\[\(\v/p^*_j\)\(1-b^*_{ij}/p^*_j\)\le C_i^*,\quad \forall j\in M,\]
with equality whenever $b^*_{ij}>0$ for $j\in M$. 
\item If firm $i$ does not exhaust its capacity, i.e., if $\sum_{j=1}^{m}b^*_{ij}<K_i$, then
\[\(\v/p^*_j\)\(1-b^*_{ij}/p^*_j\)\le 1,\quad \forall j \in M,\]
with equality whenever $b^*_{ij}>0$ for $j\in M$. 
\end{enumerate}
\end{proposition}

\begin{proof}
By differentiating $u_i\(\bbb\)$ with respect to $\b$, we obtain
\begin{align*}
\frac{\partial}{\partial \b}u_i\(\bbb\)&=\frac{\v}{p_j}-1+\b\(-\frac{\v}{p_j^2}\)=\frac{\v}{p_j}\(1-\frac{\b}{p_j}\)-1
\intertext{and}
\frac{\partial^2}{\partial b^2_{ij}}u_i\(\bbb\)&=\frac{2\v}{p_j}\(\frac{\b}{p_j}-1\)\le 0,
\end{align*}
where the last inequality follows from the fact that $\b\le p_j=\sum_{k\in N} b_{kj}$. This shows that the function $u_i\(\bbb_i;\bbb_{-i}\)$ is concave in $b_i$ for all $\bbb_{-i}$. Hence, the first order conditions \comment{Can I see here that the first order conditions are necessary and sufficient due to concavity? Because if I write the Proposition with \qt{if and only if}, I actually need to directions in the proof.} for the Nash equilibrium can be formulated as follows. 
\begin{description}[leftmargin=0cm, itemsep=5pt]
\item[Case 1:] $\sum_{j=1}^m b^*_{ij}=K_i$. The capacity constraint is satisfied with equality (tight), i.e., player $i$ exhausts its effort capacity in equilibrium. In this case, there exists a constant $C_i^*\ge 1$ such that 
\begin{tabbing}
\hspace*{1em}\= \hspace*{2em} \= \kill 
(a)\>\; if $\b>0$: $\frac{\partial}{\partial \b}u_i\(\bbb\)=C_i^*$ \\
(b)\>\; if $\b=0$: $\frac{\partial}{\partial \b}u_i\(\bbb\)\le C_i^*$.
\end{tabbing}
To see this, assume by contradiction that there exist $b^*_{ij}$ and $b^*_{ij'}$ with $\frac{\partial}{\partial \b}u_i\(\bbb^*\)> \frac{\partial}{\partial b_{ij'}}u_i\(\bbb^*\)$. Then, by using a first-order approximation, the strategy profile 
\[\bbb'_i=
\begin{cases}
b_{ik}, & \text{for }k\neq j,j'\\
\b+\delta, & \text{for }k=j\\
b_{ij'}-\delta, & \text{for }k=j'\\
\end{cases}\]
is still feasible for (NE) and yields a higher payoff to player $i$ than $\bbb^*_i$ which contradicts the fact that $\bbb_i^*$ is a best response to $\bbb^*_{-i}$. The constraint $C_i^*\ge 1$ stems from $C_i^*\ge0$ \comment{Why does $C_i^*\ge0$? } and by rescaling it with the constant $-1$ that appears in the partial derivative of $u_i$.
\item[Case 2:] $\sum_{j=1}^m b^*_{ij}<K_i$. The capacity constraint is satisfied with inequality (not tight). In this case,
\begin{tabbing}
\hspace*{1em}\= \hspace*{2em} \= \kill
(a)\>\; if $\b>0$: $\frac{\partial}{\partial \b}u_i\(\bbb\)=0$ \\
(b)\>\; if $\b=0$: $\frac{\partial}{\partial \b}u_i\(\bbb\)\le 0$.
\end{tabbing}
Since the capacity constrain is not binding, these are essentially the conditions for unconstrained maximization except for the non-negativity constrains on the $\b$'s for $i\in N,j\in M$. If not all partial derivatives are equal to zero (less or equal than zero if $\b=0$), then firm $i$ could increase its utility by increasing (decreasing) its effort $\b$ in any contest $j$ with positive (negative) partial derivative and still remain in the feasible region of (NE). \qedhere
\end{description}
\end{proof}

\subsection{Market Equilibria of Quasi-Linear Fisher Markets}\label{app:quasi-linear-ces}

We begin by stating a convex program that captures the market equilibrium prices of a quasi-linear Fisher market~\cite{Dev09,Col17}:
\begin{alignat*}{3}
\min_{p,\beta} \quad&  \sum_{j=1}^m p_j~-~\sum_{i=1}^n K_i\ln{\beta_i}\\
\text{s.t.} \quad& \v\beta_i \le p_j, ~~~\forall i\in N, j\in M,\tag{D}\label{eq:D}\\
& \beta_i \le 1,\hspace{28pt}\forall i\in N.
\end{alignat*}
To study the relationship between market equilibria (solutions of \eqref{eq:D}) and Nash equilibria (first-order conditions in Proposition \ref{prop:foc}), we consider the following equivalent problem by taking the logs of the constraints and letting $q_j:=\ln{p_j}$ and $\gamma_i:=-\ln{\beta_i}$ for all $i \in N, j\in M$.
\begin{alignat*}{3}
\min_{q,\gamma} \quad& \sum_{j=1}^m \exp(q_j)~+~\sum_{i=1}^n K_i\,\gamma_i\\
\text{s.t.} \quad& q_j+\gamma_i \ge \ln{\v},~~~\forall i\in N, j\in M,\tag{TD}\label{eq:TD}\\
&\qquad \,\gamma_i\ge 0,\hspace{28pt}\forall i\in N.
\end{alignat*}
where the constraints $\gamma_i\ge0$ correspond to the initial constraints $\beta_i\le 1$ for all $i\in N$. We can now construct the dual of \eqref{eq:TD} and gain intuition about our original problem.

\begin{lemma}\label{lem:dual}
The dual of (TD) is given by
\begin{alignat*}{3}
\max_{b,w,p} \,\,&\sum_{i=1}^n\sum_{j=1}^m\b\ln{v_{ij}}~-~\sum_{j=1}^m p_j\ln{p_j} ~-~\sum_{i=1}^n w_i\\
\text{s.t.}\,\,& \sum_{i=1}^n \b = p_j, \hspace{38.5pt}\forall j\in M,\\
& \sum_{j=1}^m \b+w_i= K_i, ~~~ \forall i\in N,\tag{SH}\label{eq:shmyrev-quasilinear}\\
& \b,w_i\ge 0,  \hspace{47pt}\forall i\in N, j\in M.
\end{alignat*}
\end{lemma} 
Variable $\b$ can be physically interpreted as the spending of player $i\in N$ on good $j\in M$. Accordingly, variable $w_i$, corresponds to the unspent budget of firm $i\in N$.
While variable $w_i$ can be eliminated via the second constraint in the convex program,
it will be useful to retain it in describing the solution of \eqref{eq:shmyrev-quasilinear} (see Proposition~\ref{prop:foc2}). 

\begin{proof}
 Let $\L\(\(q,\gamma\),\(b,p,w\)\)$ denote the Lagrange function of (TD), where $\(q,\gamma\)$ denote the primal variables and $\(b,p,w\)$ the dual variables (to be properly defined). Then, we have that
\begin{align*}
\L\(\(q,\gamma\),\(b,p,w\)\)&=\sum_{j=1}^m \exp{q_j}+\sum_{i=1}^n K_i\gamma_i+\sum_{i=1}^n\sum_{j=1}^m\b\(\ln{v_{ij}}-\gamma_i-q_j\)+\sum_{i=1}^nw_i\(-\gamma_i\)\\
&=\(-p\)^Tq+\sum_{j=1}^m\exp{q_j}+\sum_{i=1}^n\sum_{j=1}^m\b\ln{\v}+\sum_{i=1}^n(K_i-w_i-\sum_{j=1}^m\x)\gamma_i,
\end{align*}
where $p_j=\sum_{i=1}^n\b$ for all $j\in M$ and $\b,w_i\ge0$ for all $i\in N,j\in M$. To proceed with the non-linear part involving the variables $q=\(q_j\)_{j\in M}$, let $f\(q\):=\sum_{j=1}^mf_j\(q_j\)=\sum_{j=1}^m\exp{q_j}$. Minimization of $\L$ with the respect to $q$ yields 
\begin{align*}
\min_{q}\left\{\(-p\)^Tq+\sum_{j=1}^m\exp{q_j}\right\}=-\max_{q}\left\{p^Tq-\sum_{j=1}^m\exp{q_j}\right\}=-f^*\(p\),
\end{align*}
where $f^*\(q\)$ denotes the \emph{convex conjugate} of $f$, see e.g., \cite{Dev09,Boy04}. Using the separability of $f\(q\)$ in $q_j, j\in M$, we have that $f^*\(p\)=\sum_{j=1}^n f^*_j\(q_j\)$. To determine $f^*_j\(p_j\)$, observe that $\nabla f_j\(q_j\)=\exp{q_j}$ and hence that $p_j:=\exp{q_j}$, i.e., $q_j=\ln{p_j}$. This implies that 
\[f^*_j\(p_j\)=p_j\ln{p_j}-f_j\(\ln{p_j}\)=p_j\ln{p_j}-p_j.\]
Putting everything together, we obtain the convex program
\begin{alignat*}{3}
\max &~~~\sum_{i=1}^n\sum_{j=1}^m \b\ln{v_{ij}}~-&&\sum_{j=1}^m p_j\ln{p_j} +\sum_{j=1}^m p_j\\
\text{s.t.}&~~~ \sum_{i=1}^n \b = p_j,&& \quad \forall j\in M,\\
&~~~\sum_{j=1}^m \b +w_i= K_i,&& \quad \forall i\in N,\\
&~~~\b,w_i\ge 0,&&\quad \forall i\in N, j\in M.
\end{alignat*}
By summing up the first set of constraints with respect to $j\in M$ and the second set of constraints with respect to $i\in N$, we obtain that 
\[\sum_{j=1}^m p_j=\sum_{j=1}^m\sum_{i=1}^n\b =\sum_{i=1}^n K_i -\sum_{i=1}^n w_i.\]
Since $\sum_{i=1}^n K_i$ is a constant (and can, thus, be omitted from the objective function), we substitute $\sum_{j=1}^m p_j$ with $\sum_{i=1}^n w_i$ in the objective function of the (SH) to obtain the formulation in the statement of Lemma \ref{lem:dual}. This concludes the proof.
\end{proof}

Next, we derive the first-order conditions of \eqref{eq:shmyrev-quasilinear}. 

\begin{proposition}\label{prop:foc2}
At any $\bbb^\#=\(\bbb^\#_i\)_{i\in N}$ market equilibrium spending of $\Gamma$, it holds that
\begin{enumerate}[label=(\roman*),align=left]
\item If firm $i$ exhausts its capacity, i.e., if $\sum_{j=1}^{m}\b^\#=K_i$ or equivalently if $w^\#_i=0$, then there exists a constant $\tilde C_i\ge 1$ such that 
\[\v/p^\#_j\le \tilde C_i,\quad \forall j\in M\]
with equality holds whenever $\b^\#>0$ for $j\in M$. 
\item If firm $i$ does not exhaust its capacity, i.e., if $\sum_{j=1}^{m}\b^\#<K_i$ or equivalently if $w_i>0$, then
\[\v/p^\#_j\le 1,\quad \forall j \in M,\]
with equality whenever $\b^\#>0$ for $j\in M$. 
\end{enumerate}
\end{proposition}

\begin{proof}
We will show that the solution to (SH) yields a market equilibrium, by showing that the first-order conditions for the objective function of (SH)
are the same as the first order conditions for each firm's individual utility $u_i\(\bbb\)$ as given in the objective function of (NE).
From the market equilibrium perspective, the contribution $\b$ of each firm $i\in N$ is negligible when compared to the aggregate effort $p_j=\sum_{i\in N}\b$ contributed to contest $j\in M$ by all firms $i\in N$ and hence $p_j$ can be thought as constant and be independent of $\b$. This implies that differentiating $u_i\(\bbb\)$ with respect to $\b$ yields
\[\frac{\partial}{\partial \b}u_i\(\bbb\)=\frac{\v}{p_j}-1.\]
By the same reasoning as in Proposition \ref{prop:foc}, we obtain the first-order conditions in the statement of Proposition \ref{prop:foc2}.

It remains to show that the solution yields the same first order conditions and hence that its solution corresponds to a market equilibrium. To proceed, let $\phi\(\bbb,\bbw\)$ with $\(\bbb,\bbw\)=\(\(\b\)_{i\in N,j\in M},\(w_i\)_{i\in N}\)$, denote the objective function of (SH), i.e.,
\begin{align*}
\phi\(b,w\):=\,&\sum_{i=1}^n\sum_{j=1}^m \b\ln{v_{ij}}-\sum_{j=1}^m p_j\ln{p_j}-\sum_{i=1}^n w_i.
\end{align*} 
Taking the partial derivative of $\phi\(\bbb,\bbw\)$ with respect to $\b$ yields
\[\frac{\partial}{\partial \b}\phi\(\bbb,\bbw\)=\ln{\v}-\ln{p_j}-1=\ln{\frac{\v}{p_j}}-1,\]
where the fraction $\v/p_j$ denotes the marginal utility and
\[\frac{\partial^2}{\partial \b^2}\phi\(\bbb,\bbw\)=-\frac1{p_j}<0.\]
The shows that $\phi\(\bbb,\bbw\)$ is concave in $\bbb_i$ for all $\bbb_{-i}$. Hence, the first-order conditions can be now formulated as follows. 
\begin{description}[leftmargin=0cm]
\item[Case 1:] $w^\#_i=0$, or equivalently $\sum_{j=1}^m \b^\#=K_i$. In this case, at $\bbb^\#$, there exists a constant $\tilde C_i\ge 1$ such that 
\begin{tabbing}
\hspace*{1em}\= \hspace*{2em} \= \kill
(a)\>\; if $\b>0$: $\frac{\partial}{\partial \b}\phi\(\bbb,\bbw\)=\tilde C_i$\\
(b)\>\; if $\b=0$: $\frac{\partial}{\partial \b}\phi\(\bbb,\bbw\)\le \tilde C_i$.
\end{tabbing}
The reasoning is the same as in the proof of Proposition \ref{prop:foc}.
\item[Case 2:] $w^\#_i>0$, or equivalently $\sum_{j=1}^m \b^\#<K_i$. Since $\frac{\partial}{\partial w_i} \phi\(\bbb,\bbw\)=-1$, in this case, we have that
\begin{tabbing}
\hspace*{1em}\= \hspace*{2em} \= \kill
(a)\>\; if $\b>0$: $\frac{\partial}{\partial \b}\phi\(\bbb,\bbw\)=-1$ \\
(b)\>\; if $\b=0$: $\frac{\partial}{\partial \b}\phi\(\bbb,\bbw\)\le -1$.
\end{tabbing}
To see this, rewrite $\frac{\partial}{\partial w_i} \phi\(\bbb,\bbw\)=\ln{\frac11}-1$ and observe that $\bbw$ can be seen as an additional good $m+1$ with constant marginal utility $v_{i,m+1}/p_{m+1}=1$. Hence, under the assumption that $w^\#_i>0$, if there exists a good $j\in M$ with $\b^\#>0$
and $\frac{\partial}{\partial \b} \phi\(\bbb^\#,\bbw^\#\)> -1$ (resp.~$\frac{\partial}{\partial \b} \phi\(\bbb^\#,\bbw^\#\)< -1$),
then firm $i$ would be better off to reduce (resp.~increase) $w^\#_i$ and increase (resp.~decrease) $\b^\#$ by the same amount. This implies that
\[\frac{\partial}{\partial \b} \phi\(\bbb^\#,\bbw^\#\)=-1, \forall i\in N, \forall j\in M \text{ with } \b^\#>0,\]
which is equivalent to the statement in (a). Statement (b) follows by the same reasoning with the additional assumption that now $\tb_{ij}$ is on the boundary of the feasible region. \qedhere
\end{description}
\end{proof}

\subsection{Proof of \Cref{pr:me-to-ane}: Market Equilibria are Approximate Nash Equilibria}\label{app:me-to-ane}

Using the Nash and market equilibrium characterizations from \Cref{prop:foc,prop:foc2}, we can now prove \Cref{pr:me-to-ane}.

\paragraph{Step 1.} Note that $v_{ij} / p^\#_j - 1$ is the marginal utility return per unit of spending on good $j$
(the $-1$ occurs because we need to subtract the budget spent in quasi-linear utility).
At a market equilibrium $\bbp^\#$, for each firm $i$ with quasi-linear utility, it only produces goods $j$ that maximizes $v_{ij} / p^\#_j - 1$,
and its attained utility is thus $\left(\max_j \left\{ \frac{v_{ij}}{p^\#_j} \right\} - 1\right) K_i = (\beta_i-1) K_i$.

\paragraph{Step 2.} Recall that $p^\#_j = \sum_k b^\#_{kj}$. Let $i$ be any firm. Since $\b^\# \le K_i$, $\sum_{k\neq i} b^\#_{kj} = p^\#_j - \b^\# \ge p^\#_j - K_i \ge (1-\delta) p^\#_j$.

\paragraph{Step 3.} 
Now, consider the situation when all firms except $i$ have fixed their spendings at $\bbb^\#_{-i}$.
For any spending $\bbb_i$ of firm $i$, the component about good $j$ in $u_i^{\mathsf{CC}}$ is $v_{ij} \cdot \frac{\b}{\b + \sum_{k\neq i} b^\#_{kj}} - \b$.
Thus,
\[
\frac{\partial u_i^{\mathsf{CC}}}{\partial \b} = v_{ij} \cdot \frac{\sum_{k\neq i} b^\#_{kj}}{(\b + \sum_{k\neq i} b^\#_{kj})^2} - 1 \le \frac{v_{ij}}{\sum_{k\neq i} b^\#_{kj}} - 1.
\]
Due to Step 2, we have $\frac{\partial u_i^{\mathsf{CC}}}{\partial \b} \le \frac{v_{ij}}{(1-\delta) p^\#_j} - 1$ for any feasible $\bbb_i$ that satisfies $\sum_j \b \le K_i$.

\paragraph{Step 4.} 
Due to Step 3, when all firms except $i$ have fixed their spendings at $\bbb^\#_{-i}$, no matter how firm $i$ spends its budget, its maximum possible utility is either zero (when its spends nothing),
or $\max_j \left\{ \frac{v_{ij}}{(1-\delta) p^\#_j} - 1 \right\}\cdot K_i = \left(\frac{\beta_i}{1-\delta} - 1\right)K_i$.

\paragraph{Step 5.}
If $\beta_i > 1$, then by comparing the results of Step 1 and Step 4, firm $i$ can improve its utility by a ratio of at most $\frac{\frac{\beta_i}{1-\delta}-1}{\beta-1}$.

If $\beta_i < 1$ and $\frac{\beta_i}{1-\delta} \le 1$, then firm $i$ spends nothing at both market equilibrium and in the situation depicted by Steps 3 and 4, so its utility cannot be improved.

However, if $\beta_i < 1$ but $\frac{\beta_i}{1-\delta} > 1$, then firm $i$ spends nothing at market equilibrium, thus attaining utility zero, but it will spend some money in the situation depicted by Steps 3 and 4, yielding a positive utility. This forces $\bbb^\#$ not be a $\delta'$-NE for any finite $\delta'$.

Summarizing the three cases discussed in Step 5, Proposition~\ref{pr:me-to-ane} follows.
\section{Missing Proofs in Section~\ref{sect:PR}}\label{app:PR}

The high-level structure of the proof of \Cref{thm:pr-qlin} was already given in \Cref{sect:PR}. In \Cref{app:quasi-linear-ces}, we showed that the convex program \eqref{eq:shmyrev-quasilinear}
captures the market equilibrium spending of a quasi-linear Fisher market.
Thus, there are two remaining components of the proof which we will complete here: proving \Cref{lem:one-Bregman-convex} and deriving
the \textsf{PR-QLIN} Algorithm from the Mirror Descent procotol in \Cref{alg:md}.

\subsection{Proof of \Cref{lem:one-Bregman-convex}}

Recall that the function $F$ with domain $(\bbb,\bbw,\bbp)$ is
\begin{align*}
F(\bbb,\bbw,\bbp):=-\sum_{i=1}^n\sum_{j=1}^m\b\ln{v_{ij}} +\sum_{i=1}^n w_i +\sum_{j=1}^m p_j\ln{p_j}.
\end{align*}
As we have discussed in \Cref{sect:PR}, we can eliminate $\bbp,\bbw$ since they are functions of $\bbb$,
via the equalities $\sum_i \b = p_j$ for each $j$ and $K_i - \sum_j \b = w_i$ for each $i$. In the proof below, we will keep the notation $\bbw,\bbp$ as doing so will help to ease the clustering of algebra (however, we keep in mind that $\bbp,\bbw$ are now functions of $\bbb$ rather than independent variables).
To prove \Cref{lem:one-Bregman-convex}, it suffices to show
\begin{equation}\label{eq:ast}
0\le F(\bbb') - F(\bbb) - \inner{\nabla F(\bbb)}{\bbb'-\bbb}\le \KL\(\bbb' \| \bbb \).
\end{equation}
Note that
\[
\frac{\partial F}{\partial \b} ~=~ \ln p_j - \ln \v.
\]
To proceed, let for convenience $\Phi(\bbb,\bbb'):= F(\bbb') - F(\bbb) - \inner{\nabla F(\bbb)}{\bbb'-\bbb}$. Then we have
\begin{align*}
\Phi(\bbb,\bbb') &=-\sum_{i=1}^n \sum_{j=1}^m \(\b' - \b\)\ln{\v}+ \sum_{j=1}^m \( p_j'\ln p_j' - p_j \ln p_j \) + \sum_{i=1}^n (w'_i - w_i) \\
&\phantom{=\,\,} -\sum_{i=1}^n \sum_{j=1}^m (\b' - \b)(\ln p_j - \ln{\v})\\
&=\sum_{j=1}^m \( p_j'\ln p_j' - p_j \ln p_j \)+ \sum_{i=1}^n \( (K_i - \sum_{j=1}^m \b') - (K_i - \sum_{j=1}^m \b) \)\\
&\phantom{=\,\,}-\sum_{j=1}^m (\ln p_j) \sum_{i=1}^n (\b' - \b).
\end{align*}
Recall that $\sum_i \b = p_j$, $\sum_i \b' = p'_j$. Thus,
\begin{align*}
\Phi(\bbb,\bbb') &= \sum_{j=1}^m \( p_j'\ln p_j' - p_j \ln p_j \) + \sum_{j=1}^m (p_j - p'_j)-\sum_{j=1}^m (\ln p_j)(p'_j - p_j)\\
&= \sum_{j=1}^m p'_j \ln \frac{p'_j}{p_j} + \sum_{j=1}^m (p_j - p'_j)= \KL(\bbp'\|\bbp).
\end{align*}
Since $\KL$ is always non-negative, this shows the first inequality in equation \eqref{eq:ast}.
The second inequality in equation \eqref{eq:ast} is also straightforward since $\bbb$ is a \emph{refinement} of $\bbp$
and $\bbb'$ is a refinement of $\bbp'$, so due to a well-known property of KL divergence (see the proof of Lemma 7 in \cite{Bir11}),
\[
\KL\(\bbp'\| \bbp\) ~\le~ \KL\(\bbb'\| \bbb\),
\]
which concludes the proof. 

\subsection{Deriving the \textsf{PR-QLIN} Algorithm from Mirror Descent}

\newcommand{\of}{\overline{F}}

The Mirror Descent update rule for the objective function $F$ in \eqref{eq:shmyrev-quasilinear} is (cf. \Cref{alg:md})
\begin{align*}
(\bbb^{t+1},\bbw^{t+1}) = \, & \argmin_{(\bbb,\bbw)\in C} \left\{\sum_{i=1}^n (w_i - w_i^t)+\sum_{i=1}^n \sum_{j=1}^m \(\b - \b^t\)\(1 - \ln{\frac{\v}{p_j^t}}\)+\sum_{i=1}^n\KL(\bbb_i \| \bbb_i^t) \right\}.
\end{align*}
Since $\sum_{j=1}^m \b + w_i$ is a constant in the domain $C$, we may ignore any term that does not depend on $\bbb$ and $\bbw$, and any positive constant factor in the objective function and simplify the above update rule to
\begin{align*}
&(\bbb^{t+1},\bbw^{t+1}) =\argmin_{(\bbb,\bbw)\in C} \left\{ -\sum_{i=1}^n \sum_{j=1}^m \b\left(\ln{\frac{\v}{p_j^t}} - \ln \frac{\b}{\b^t}+ 1\right)\right\} =: \argmin_{(\bbb,\bbw)\in C}\of(\bbb,\bbw).
\end{align*}
We have that $\frac{\partial}{\partial w_i} \of(\bbb,\bbw)=0$ and 
\begin{align*}
\frac{\partial}{\partial \b} \of(\bbb,\bbw)= \ln {\frac{\b}{\b^t}} - \ln {\frac{\v}{p_j^t}}\,.
\end{align*}
As before, for each fixed $i$, the values of $\ln \frac{\b}{\b^t}-\ln \frac{\v }{p_j^t}$ for all $j$ are identical. In other words, there exists $c_i > 0$ such that
\[\b = c_i \cdot \frac{\v\b^t}{p_j^t}\, .\]
There are two cases which depend on $S_i^t:= \sum_{j=1}^m \frac{\v\b^t}{p_j^t}$ as follows
\begin{itemize}[leftmargin=*]
\item If $S^t_i \ge K_i$, then for each $j$ we set $\b^{t+1} = K_i \cdot \frac{\v \b^t}{p_j^t} / S^t_i$, and $w_i^{t+1} = 0$.
At this point, we have $\frac{\partial}{\partial \b} \of(\bbb,\bbw) = \ln \frac{K_i}{S^t_i} \le 0=\frac{\partial}{\partial w_i} \of(\bbb,\bbw)$, so the optimality condition is satisfied.
\item if $S^t_i < K_i$, then for each $j$, we set $\b^{t+1} = \frac{\v \b^t}{p_j^t}$,
and $w^{t+1} = K_i - \sum_{j=1}^m \b^{t+1} > 0$.
At this point, $\frac{\partial}{\partial \b} \of(\bbb,\bbw) = 0 =\frac{\partial}{\partial w_i} \of(\bbb,\bbw)$, so the optimality condition is again satisfied.
\end{itemize}
\section{Missing Proofs in Section~\ref{sect:liyorke}}\label{app:liyorke}

\begin{lemma}\label{lem:invariant}
For any $\alpha\ge1$, let $f\(x\)=x+\eta\(\frac{1}{4x}-\alpha\)$. Then, for any $\eta \in \lt 3/4\alpha^2,1/\alpha^2\)$, the interval $I\(\eta\):=\lt L\(\eta\), U\(\eta\)\rt$, with
\begin{align*}
L\(\eta\)&=\sqrt{\eta}\(1-\alpha\sqrt{\eta}\), \\
U\(\eta\)&=\frac{\sqrt{\eta}}{4\(1-\alpha\sqrt{\eta}\)}\cdot\(5-12\alpha\sqrt{\eta}+8\alpha^2\eta\)
\end{align*}
is invariant under $f$, i.e., $f\(I\(\eta\)\)\subseteq  I\(\eta\)$, and $x^*=1/4\alpha \in I\(\eta\)$ for all $\eta$.
\end{lemma}

\begin{proof}[Proof of Lemma \ref{lem:invariant}.]
The first and second derivatives of $f\(x\)$ are $f'\(x\)=1-\frac{\eta}{4x^2}$ with $f'\(x\)=0$, iff $x=\sqrt{\eta}/2$ and $f''\(x\)=\eta/2x^3>0$ for any $x>0$. Hence, $f$ is concave and attains its global minimum at $x=\sqrt{\eta}/2$, with 
\[f\(\sqrt{\eta}/2\)=\sqrt{\eta}\(1-\alpha\sqrt{\eta}\).\]
This is precisely the lower bound $L\(\eta\)$ of the interval $I\(\eta\)$. Observe that since $\eta \in[3/4\alpha^2,1/\alpha^2)$ by assumption, we also have that 
\[\frac{\sqrt{\eta}}2\ge \frac{\sqrt{3}/2\alpha}{2}=\sqrt{3}\cdot\frac{1}{4\alpha}>\frac1{4\alpha}\,,\]
which implies that the fixed point $x^*=1/4\alpha$ is less than the argmin $x=\sqrt{\eta}/2$ of $f$. By taking the derivative with respect to $\eta$, $L\(\eta\)$ is decreasing for any $\eta\in[3/4\alpha^2,1/\alpha^2)$, and hence, we have that
\begin{equation}\label{eq:lower}
L\(\eta\)\le L\(3/4\alpha^2\)=\frac{\sqrt{3}}{2\alpha}\(1-\alpha\cdot\frac{\sqrt{3}}{2\alpha}\)<\frac{1}{4\alpha}\,,
\end{equation}
for all $\eta$ in this range. By the observation above, this also implies that $L\(\eta\)<\sqrt{\eta}/2$. To obtain the upper bound $U\(\eta\)$, we need to evaluate $f$ at $f\(\sqrt{\eta}/2\)$, since $f$ is (steeply) decreasing for small values. Hence, after some algebraic manipulations, we get
\begin{align*}
f\(\sqrt{\eta}\(1-\alpha\sqrt{\eta}\)\)&=\sqrt{\eta}\(1-\alpha\sqrt{\eta}\)+\eta\(\frac{1}{4\sqrt{\eta}\(1-\alpha\sqrt{\eta}\)}-\alpha\)=\frac{\sqrt{\eta}\cdot\(5-12\alpha\sqrt{\eta}+8\alpha^2\eta\)}{4\(1-\alpha\sqrt{\eta}\)}\,.
\end{align*}
Let $U\(\eta\):=f\(\sqrt{\eta}\(1-\alpha\sqrt{\eta}\)\)$. Next, we observe that $U\(\eta\)>\sqrt{\eta}/2$ for any $\eta \in[3/4\alpha^2,1/\alpha^2)$, which -- by the observation above, that $1/4\alpha<\sqrt{\eta}/2$ in this range -- also implies that $1/4\alpha<U\(\eta\)$. To show that 
\begin{equation}\label{eq:upper}
U\(\eta\)=\frac{\sqrt{\eta}}{4\(1-\alpha\sqrt{\eta}\)}\cdot\(5-12\alpha\sqrt{\eta}+8\alpha^2\eta\)>\frac{\sqrt{\eta}}{2}\,,
\end{equation}
we note that $1-\alpha\sqrt{\eta}>0$ for all $\eta$ in this range, and hence, after some manipulations, this is equivalent to 
\[8\alpha^2\eta-10\alpha\sqrt{\eta}+3>0.\]
The roots of the expression on the left side of the inequality are $\sqrt{\eta}=\frac{3}{4\alpha}$ and $\sqrt{\eta}=\frac1{2\alpha}$. Hence, for $\eta>\frac{9}{16\alpha^2}$, in particular for $\eta\in[3/4\alpha^2,1/\alpha^2)$, the above inequality holds and hence, we have that $U\(\eta\)>\sqrt{\eta}/2$. Summing up, equations \eqref{eq:lower} and \eqref{eq:upper} establish that
\[L\(\eta\)<\frac{1}{4\alpha}<\frac{\sqrt{\eta}}{2}<U\(\eta\)\,,\]
for all $\eta\in[3/4\alpha^2,1/\alpha^2)$. It remains to show that $f\(I\(\eta\)\)\subseteq I\(\eta\)$. This indeed follows by the construction of $I\(\eta\)$. Specifically, $x \in [L\(\eta\),\sqrt{\eta}/2]$ implies that 
\[f\(x\)\in[f\(\sqrt{\eta}/2\),f\(L\(\eta\)\)]=[L\(\eta\),U\(\eta\)]\,,\]
since $f$ is (strictly) decreasing in $[L\(\eta\),\sqrt{\eta}/2]$ (cf. $f'$ above). Similarly, $x \in [\sqrt{\eta}/2,U\(\eta\)]$ implies that
\[f\(x\)\in[f\(\sqrt{\eta}/2\),f\(U\(\eta\)\)]=[L\(\eta\),f\(U\(\eta\)\)]\,,\]
since $f$ is (strictly) increasing in $[\sqrt{\eta}/2,U\(\eta\)]$. Now, since $U\(\eta\)>1/4\alpha$, where $1/4\alpha$ is the unique fixed point of $f$ and since $1/4\alpha$ lies at the part at which $f$ is decreasing, it follows that
\[U\(\eta\)>f\(U\(\eta\)\)\,,\]
and hence that $[L\(\eta\),f\(U\(\eta\)\)]\subset [L\(\eta\),U\(\eta\)]=I\(\eta\)$ which completes the proof.
\end{proof}

\begin{figure}[!htb]
\centering
\includegraphics[width=0.5\linewidth]{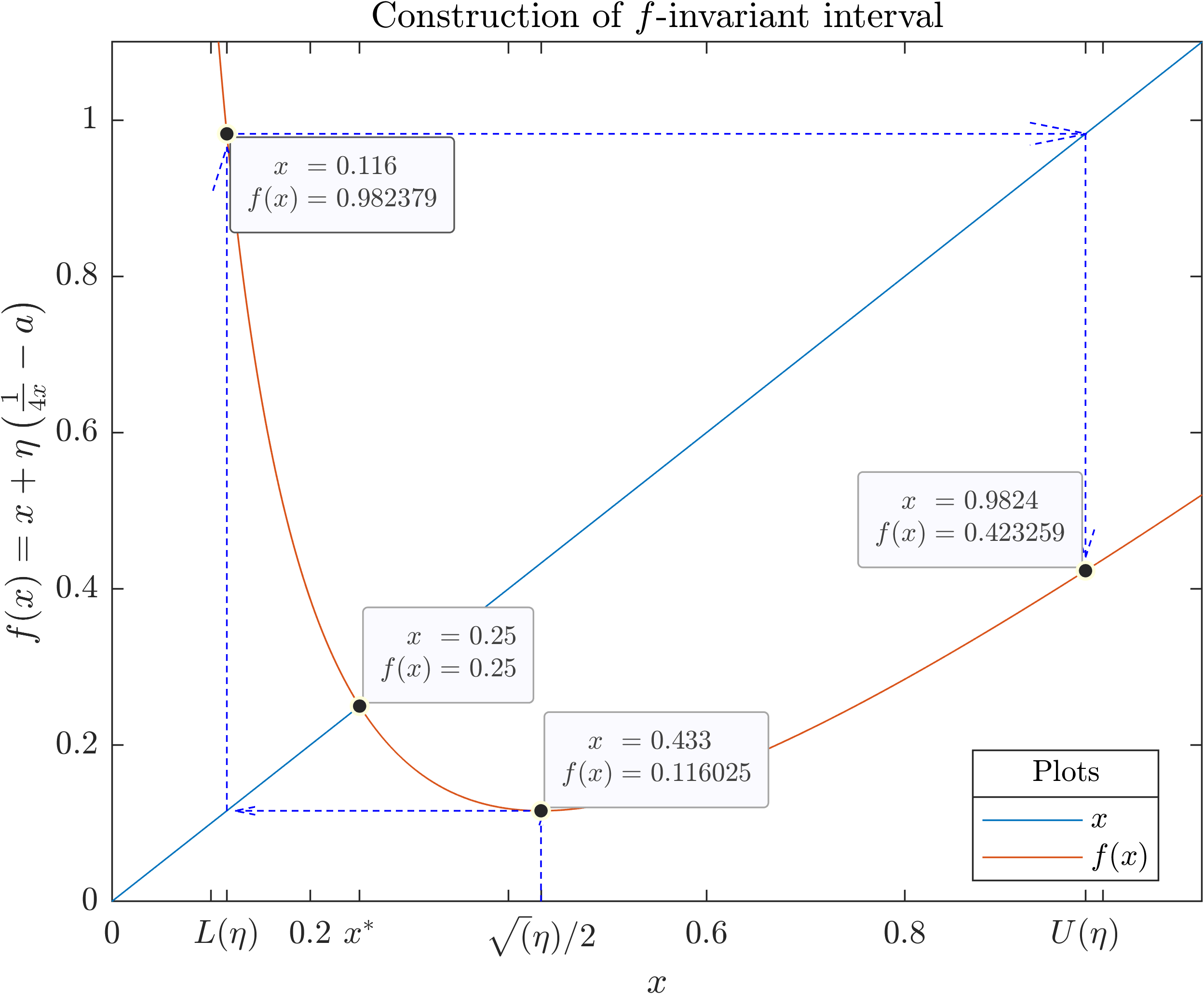}
\caption{Sketch of the construction of the invariant interval of $f$ for $\alpha=1$ and $\eta=3/4$ in the proof of Lemma \ref{lem:invariant}. Starting from the argmin of $f$ at $x=\sqrt{\eta}/2$ and following the blue dashed arrows, we get the interval $I\(\eta\)=[L\(\eta\),U\(\eta\)]$ with $f\(I\(\eta\)\)\subseteq I\(\eta\)$ and $x^*=1/4\alpha\in I\(\eta\)$.}
\label{fig:invariant}
\end{figure}

\begin{remark}\label{rem:invariant}
The selected interval, $\eta\in[3/4\alpha^2,1/\alpha^2)$, in the statement of Lemma \ref{lem:invariant} is not minimal, in the sense that Li-Yorke chaos also appears for values of $\eta$ outside these bounds. However, it is sufficient for our purposes. A sketch of the constructive proof of Lemma \ref{lem:invariant} is given in Figure \ref{fig:invariant}. In the depicted instantiation, $\alpha=1$ and $\eta=3/4$, yet the image is qualitatively the same for all values of $\eta$ in the above range. 
\end{remark}

Next, we turn to the existence of a fixed point $p\neq x^*\in I\(\eta\)$ of $\fof:=\(f\circ f\circ f\)\(x\)$. The result is formally established in \Cref{lem:periodic_3}, which concludes the two-step proof of Li-Yorke chaos.

\begin{proposition}\label{lem:periodic_3}
Let $\alpha\ge1$ arbitrary and let $\fof:=\(f\circ f\circ f\)\(x\)$. Then, for any $\eta \in \lt 3/4\alpha^2,1/\alpha^2\)$, $\fof$ is continuous for $x>0$ and  
\begin{align*}
f^{\(3\)}\(\sqrt{\eta}/2\)-\sqrt{\eta}/2 \ge 0 > f^{\(3\)}\(U\(\eta\)\)-U\(\eta\),
\end{align*}
where $U\(\eta\)$ is defined in Lemma \ref{lem:invariant}. In particular, $\fof$ has an additional fixed point $p\neq x^*=1/4\alpha$ with $p\in \lt \sqrt{\eta}/2, U\(\eta\)\rt$. 
\end{proposition}

\begin{proof}
Since $\eta\in \lt 3/4\alpha^2,1/\alpha^2\)$, it will be convenient to parametrize $\eta$ as
\[
\eta:=\frac{c}{4\alpha^2},\vspace*{0.1in}
\]
with $c\in[3,4)$. We present here the proof for $c\in[3+\delta,4)$ for $\delta\approx 0.00918$ which is more intuitive and omit the case $c\in[3,3+\delta)$. The argument is essentially the same for values of $\eta \in [3/4\alpha^2,(3+\delta)/4\alpha^2]$. However, the numerical evaluation of $\fof$ is more involved and is thus omitted. In any case, the exact interval for which existence of an additional fixed point is established does not affect the overall result.\par

We first prove that $f^{\(3\)}\(\sqrt{\eta}/2\)-\sqrt{\eta}/2 \ge 0$. The proof exploits the fact that $f^{\(3\)}\(x\)$ is locally maximized at the argmin $x=\sqrt{\eta}/2$ of $f\(x\)$. Hence, (using mathematical software), we obtain that 
\begin{align*}
f^{\(3\)}\(\sqrt{\eta}/2\)-\sqrt{\eta}/2=\frac1{2\alpha}\cdot \frac{64 c^{3/2}+6 c^{5/2}-32 c^2-57 c+19 \sqrt{c}}{-4c^{3/2}+20c-34\sqrt{c}+20}\,,
\end{align*}
which can be shown to be increasing in $c$. A visual representation of $f^{\(3\)}\(\sqrt{\eta}/2\)-\sqrt{\eta}/2$ is given in the left panel of Figure \ref{fig:proof}. Since equality with zero occurs for $c\approx 3.00918$, it follows that the expression remains positive for $c\in[3.00918,4]$, or equivalently for $\eta\in [3.00918/4\alpha^2,1/\alpha^2]$.\par
Next, we prove that $f^{\(3\)}\(U\(\eta\)\)-U\(\eta\) < 0$. Recall by the proof of Lemma \ref{lem:invariant} that $U\(\eta\)=f^{\(2\)}\(\sqrt{\eta}/2\)$, where $f^{\(2\)}\(x\):=\(f\circ f\)\(x\)$.  Hence, 
\[f^{\(3\)}\(U\(\eta\)\)-U\(\eta\)=f^{\(5\)}\(\sqrt{\eta}/2\)-f^{\(2\)}\(\sqrt{\eta}/2\)\,,\]
which can be evaluated again by mathematical software. The result is 
\begin{align*}
&f^{\(3\)}\(U\(\eta\)\)-U\(\eta\)=-\frac1{2\alpha}\cdot\frac{\(\sqrt{c}-1\)^3 \sqrt{c} \(-6 c^{3/2}+26 c-38 \sqrt{c}+19\)^2 }{\left(2 c-6 \sqrt{c}+5\right) \left(-34 c^{3/2}+6 c^2+74 c-74 \sqrt{c}+29\right)}\cdot\\
&\cdot \frac{\(-1268 c^{3/2}-196 c^{5/2}+24 c^3+677 c^2+1361 c-796 \sqrt{c}+199\)}{\(-15592 c^{3/2}-7486 c^{5/2}-536 c^{7/2}+48 c^4+2638 c^3+13420 c^2+11490 c-4922 \sqrt{c}+941\)}\,.
\end{align*}
As above, the right hand side can be shown to be decreasing in $c$, cf. right panel of Figure \ref{fig:proof}. Again, equality with zero occurs $c\approx 3.00918$ and hence the expression remains negative for negative for any $c\in[3.00918,4]$, or equivalently for $\eta\in\(3.00918/4\alpha^2,1/\alpha^2\)$. The existence follows from the continuity of $\fof$ in this interval -- which can be shown by a standard exercise \comment{Do I need to show this rigorously?} for values of $\eta<1/\alpha^2$ -- and the intermediate value theorem. This proves the statement of the Proposition.
\begin{figure}[!tb]
\centering
\includegraphics[width=0.49\linewidth]{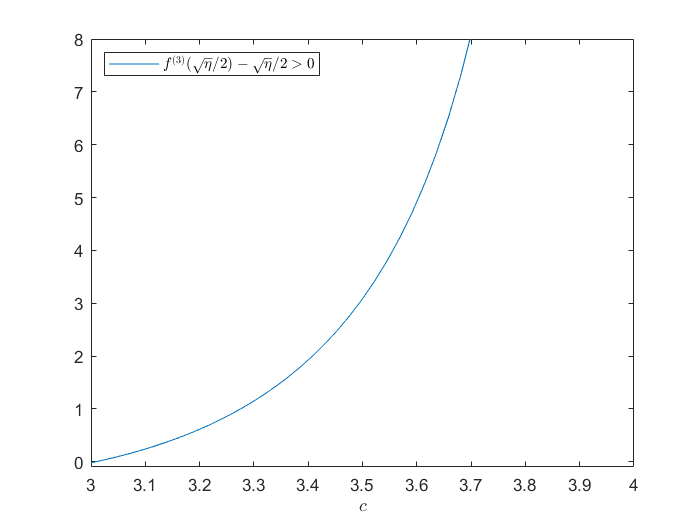}\hfill
\includegraphics[width=0.49\linewidth]{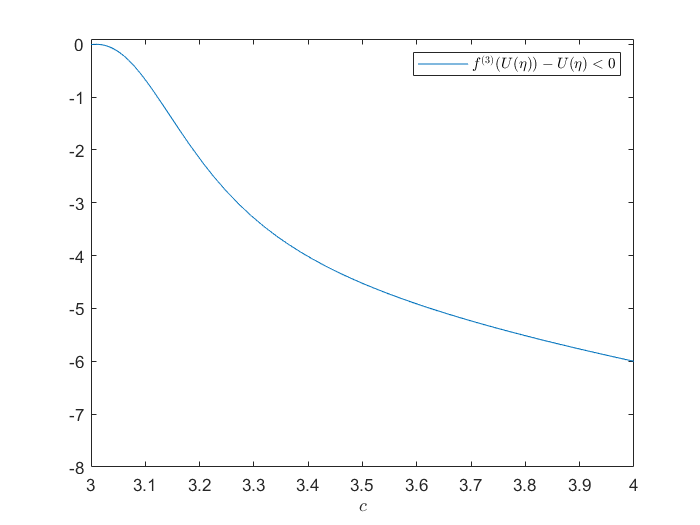}
\caption{The graphs of $f^{\(3\)}\(\sqrt{\eta}/2\)-\sqrt{\eta}/2$ (left panel) and $f^{\(3\)}\(U\(\eta\)\)-U\(\eta\)$ (right panel) for values of $c[3,4]$ or equivalently for $\eta\in \lt 3/4\alpha^2,1/\alpha^2\)$. The first expression is increasing and remains positive for any $c\ge 3.00918$ and the second expression is decreasing and remains negative for any $c\ge3.00918$.}  
\label{fig:proof}
\end{figure}
\end{proof}

The statement of \Cref{lem:periodic_3} is illustrated in Figure \ref{fig:periodic_3}.
\begin{figure}[!tb]
\centering
\includegraphics[width=0.49\linewidth]{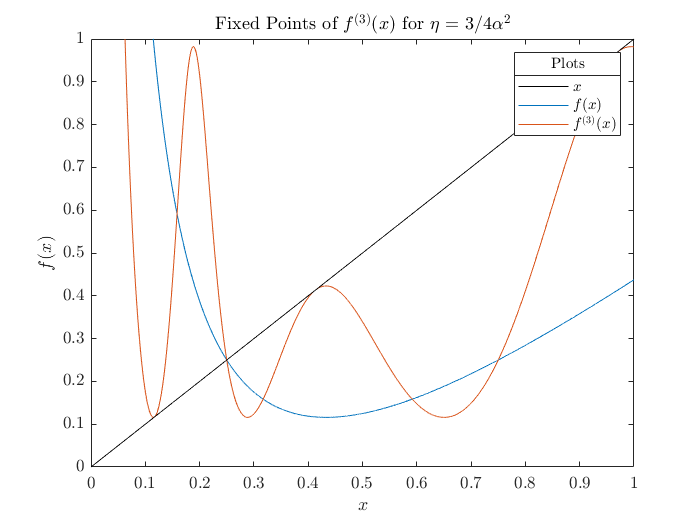}\hfill
\includegraphics[width=0.49\linewidth]{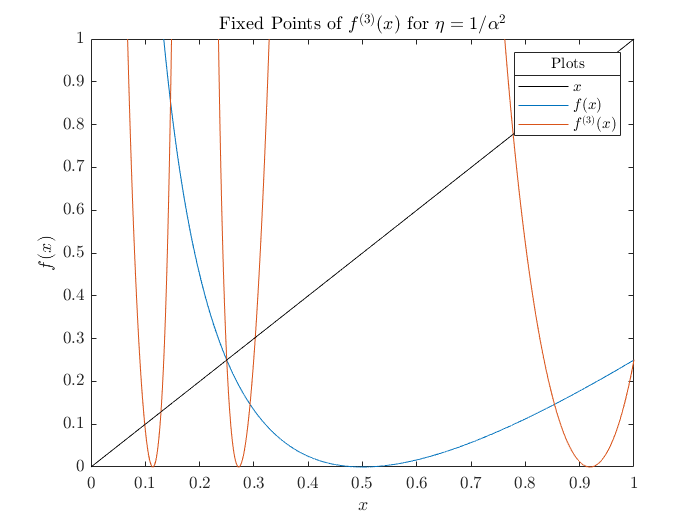}
\caption{Points of period $3$, i.e., fixed points of the function $\fof=\(f\circ f\circ f\)\(x\)$ in the interval $\(0,1\)$. The plots of $f\(x\)$ and $\fof$ have one common fixed point (common intersection point with the diagonal), however $\fof$ has more fixed points. For intermediate cases, i.e., for $3/4\alpha^2<\eta<1/\alpha^2$, the graph of $\fof$ stretches continuously from the image in the left panel to the image in the right panel. For the illustrations $\alpha=1$, but the outcome is qualitatively equivalent for any choice of $\alpha$ in the admissible range.}
\label{fig:periodic_3}
\end{figure}

\subsection{Best Response Dynamics}\label{app:br}

While the previous technique does not allow us to establish Li-Yorke chaos in the case of Best Response (BR) Dynamics, we complement the existing empirical results of \cite{Puu91} and \cite{War18} with a formal proof of the stability properties of the BR dynamics via eigenvalue analysis of a first order linear approximation of the original non-linear system. Due to the first order approximation, the boundary cases are not precise, however, qualitatively, the result is robust. Recall, that for TC with isoelastic demand, the BR dynamics take the form 
\begin{equation}\label{eq:br_dynamics}
x_i^{t+1}\la\sqrt{x_{-i}^t/\alpha_i}-x_{-i}^t, \,\, \text{for } i=1,2 \text{ and }t\in \mathbb N.
\end{equation}

\begin{proposition}\label{prop:stability}
Under the Best Response (BR) update rule, \eqref{eq:br_dynamics}, the evolution of the sequence $\(x^t,y^t\)$ of effort level pairs around the unique equilibrium depends on the value of the ratio $r:=\alpha/\beta$ of the abilities of the two firms relative to the points $r_0:=3+2\sqrt{2}$ and $r_0^{-1}$
\begin{itemize}[noitemsep]
\item If $r_0^{-1}<r<r_0$, the effort levels are spiralling inwards towards the unique equilibrium (stable spiral focus). 
\item If $r=r_0^{-1}$ or $r=r_0$, the effort levels cycle around the unique equilibrium effort level (neutral center). 
\item If $r<r_0^{-1}$ or $r>r_0$, the effort levels are spiralling outwards away from the unique equilibrium effort level (unstable spiral focus).\vspace*{0.08in}
\end{itemize}
\end{proposition}

An instantiation of the result of \Cref{prop:stability} is shown in \Cref{fig:spiral_2}.

\begin{proof}[Proof]
The unique resting point of the best response dynamics is given by the simultaneous solutions of the equations in equation \eqref{eq:br_dynamics} which yields 
\[\(x^*,y^*\)=\(v\beta,v\alpha\)/\(\alpha+\beta\)^2.\]
The stability of the dynamical system \eqref{eq:br_dynamics} around the equilibrium $\(x^*,y^*\)$ can be determined by the eigenvalues of the Jacobian of the system at the equilibrium point, 
\begin{align*}
J&=\left.\begin{pmatrix} \frac{\partial}{\partial x}x^{t+1}\(x,y\) & \frac{\partial}{\partial y}x^{t+1}\(x,y\) \\ \frac{\partial}{\partial x}y^{t+1}\(x,y\) & \frac{\partial}{\partial y}y^{t+1}\(x,y\)\end{pmatrix}\right|_{\(x,y\)=\(x^*,y^*\)}=\begin{pmatrix}0 & -\frac1{2r}\(r-1\)\\ \frac12\(r-1\) & 0\end{pmatrix},
\end{align*}
which yields the characteristic equation 
\[\det\(J-\lambda I\)=\lambda^2+\frac{1}{4r}\(r-1\)^2,\]
and hence, the complex conjugate eigenvalues $\lambda_{1,2}$ with absolute value $|\lambda_1|=|\lambda_2|=\frac{1}{2\sqrt{r}}|r-1|$. This implies that the dynamics oscillate around the equilibrium with the exact behavior depending on whether the absolute value of $\lambda_{1,2}$ is less than, equal to or larger than $1$. In particular, 
\begin{align*}
\frac{1}{2\sqrt{r}}|r-1|=1&\iff r^2-6r+1=0 \iff \(r-r_0\)\(r-r_0^{-1}\)=0,
\end{align*}
with $r_0:=3+2\sqrt{2}$, which concludes the proof. 
\end{proof}

\begin{figure}[!htb]
\centering
\includegraphics[scale=0.9]{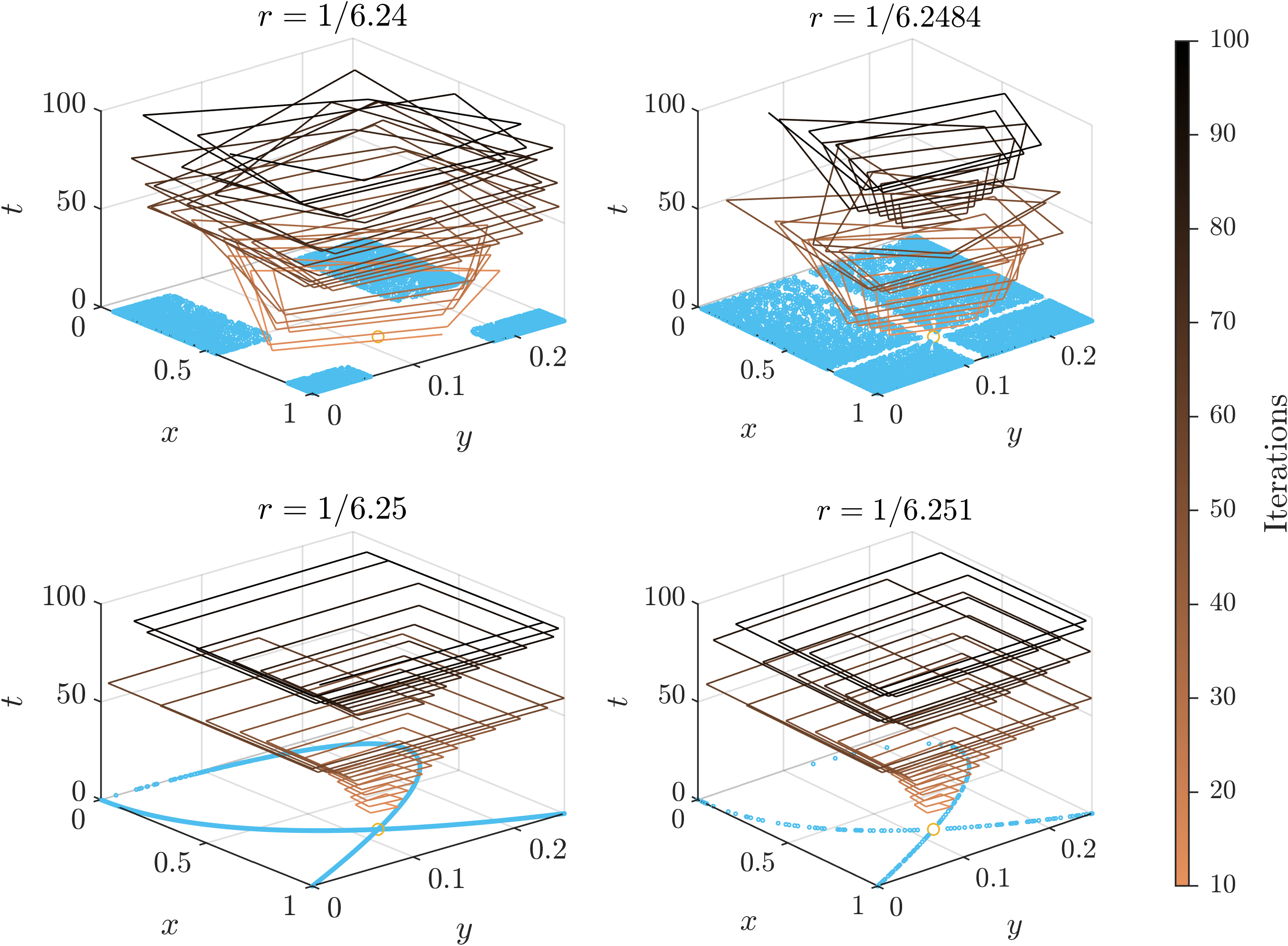}
\caption{The trajectories (light to dark lines) of outputs $\(x^t,y^t\)$ with respect to time $t\in[10,100]$ (vertical axis). The light blue dots denote their projections on the $x-y$ plane for a longer time span $t\in [10^3,2\cdot10^3]$. While all dynamics spiral away from the unique equilibrium (red circle on the $x-y$ plane) in agreement with Proposition \ref{prop:stability}, their attractors change significantly for only small perturbations in the parameters of the game as expressed by the asymmetry-ratio $r$.}
\label{fig:spiral_2}
\end{figure}

\end{document}